\documentclass[11pt,letterpaper]{article}
\usepackage[margin=1in]{geometry}
\usepackage{clrscode3e, amsmath, amsthm, amssymb}
\usepackage{graphicx}
\usepackage{subfig}
\usepackage{enumitem,lipsum}
\usepackage{url}

\usepackage{caption}
\usepackage{ifpdf}
\ifpdf\fi

\newtheorem{theorem}{Theorem}[section]
\newtheorem{lemma}[theorem]{Lemma}

\newtheorem{definition}[theorem]{Definition}

\newcommand{\eat}[1]{}

\newcommand{\comment}[1]{}

\long\def\skipthis#1{}

\def\squareforqed{\hbox{\rlap{$\sqcap$}$\sqcup$}}
\def\qed{\ifmmode\else\unskip\quad\fi\squareforqed}
\def\smartqed{\def\qed{\ifmmode\squareforqed\else{\unskip\nobreak\hfil
\penalty50\hskip1em\null\nobreak\hfil\squareforqed
\parfillskip=0pt\finalhyphendemerits=0\endgraf}\fi}}

\smartqed

\title{Hierarchical Heavy Hitters with the Space Saving Algorithm}

\title{Hierarchical Heavy Hitters with the Space Saving Algorithm}

\author{
Michael~Mitzenmacher \and
Thomas Steinke \and
Justin Thaler\\
School of Engineering and Applied Sciences\\
Harvard University,
Cambridge, MA 02138\\ Email: \{jthaler, tsteinke\}@seas.harvard.edu \, \, michaelm@eecs.harvard.edu} 
\date{}

\begin{document}
\maketitle

\begin{abstract} 
The Hierarchical Heavy Hitters problem extends the notion of frequent items to data arranged in a hierarchy. This problem has applications to network traffic monitoring, anomaly detection, and  DDoS detection. We present a new streaming approximation algorithm  for computing Hierarchical Heavy Hitters that has several advantages over previous algorithms. It improves on the worst-case time and space bounds of earlier algorithms, is conceptually simple and substantially easier to implement, offers improved accuracy guarantees, is easily adopted to a distributed or parallel setting, and can be efficiently implemented in commodity hardware such as ternary content addressable memory (TCAMs). We present experimental results showing that for parameters of primary practical interest, our two-dimensional algorithm is superior to existing algorithms in terms of speed and accuracy, and competitive in terms of space, while our one-dimensional algorithm is 
also superior in terms of speed and accuracy for a more limited range of parameters.



\end{abstract}

\section{Introduction}

Finding \textit{heavy hitters}, or frequent items, is a fundamental
problem in the data streaming paradigm.  As a practical motivation,
network managers often wish to determine which IP addresses are
sending or receiving the most traffic, in order to detect anomalous
activity or optimize performance.  Often, the large volume of network
traffic makes it infeasible to store the relevant data in
memory. Instead, we can use a \textit{streaming} algorithm to compute (approximate)
statistics in real time given sequential access to the data and using
space sublinear in both the universe size and stream length.

We present and analyze a streaming approximation algorithm for a
generalization of the Heavy Hitters problem, known as Hierarchical
Heavy Hitters (HHHs).
The definition of HHHs is motivated by the observation that some data
are naturally hierarchical, and
ignoring this when tracking frequent items may mean the loss
of useful information.  Returning to our example of IP addresses,
suppose that a single entity controls all IP addresses of the subnet
021.132.145.*, where * is a wildcard byte. It is possible for the
controlling entity to spread out traffic uniformly among this set of
IP addresses, so that no single IP address within the set of addresses
021.132.145.* is a heavy hitter.  Nonetheless, a network manager may
want to know if the sum of the traffic of all IP addresses in the
subnet exceeds a specified threshold.  

One can expand the concept further to consider multidimensional
hierarchical data.  For example, one might track traffic between
source-destination \emph{pairs} of IP addresses at the router level.  In that
case, the network manager may want to know if there is a Heavy Hitter
for network traffic at the level of two IP addresses, between a source
IP address and a destination subnet, between a source subnet and a
destination IP address, or between two subnets. This motivates the
study of the \emph{two-dimensional} HHH problem.

There is some subtlety in the appropriate
definitions, as it makes sense to require that an element is not marked
as an HHH simply because it has a significant descendant, but because
the \textit{aggregation} of its children makes it significant;
otherwise, the algorithm returns redundant, less helpful information.
We present the definitions shortly, following previous
work that has explored HHHs for both one-dimensional and 
multi-dimensional hierarchies \cite{seminal, pods, conference, journal, esv, china, france}.

HHHs have many applications, and have been central to proposals for
real-time anomaly detection \cite{anomaly} and DDos detection
\cite{DDos}. While IP addresses serve as our motivating example
throughout the paper, our algorithm applies to arbitrary hierarchical
data such as geographic or temporal data. We demonstrate that our
algorithm has several advantages, combining improved worst-case
time and space bounds with more practical advantages such as simplicity,
parallelizability, and superior performance on real-world data.

Our algorithm utilizes the Space Saving algorithm, proposed by
Metwally et al. \cite{2005}, as a subroutine.  Space Saving is a
\emph{counter-based} algorithm for estimating item frequencies,
meaning the algorithm tracks a subset of items from the universe,
maintaining an approximate count for each item in the subset.
Specifically, the algorithm input is a stream of pairs $(i,c)$ where
$i$ is an item and $c > 0$ is a frequency increment for that item.  At
each time step the algorithm tracks a set $T$ of items, each with a
counter.  If the next item $i$ in the stream is in $T$, its counter is
updated appropriately.  Otherwise, the item with the smallest counter
in $T$ is removed and replaced by $i$, and the counter for $i$ is set
to the counter value of the item replaced, plus $c$.  This
approach for replacing items in the set may seem counterintuitive,
as the item $i$ may have an exaggerated count after placement, but the result
is that if $T$ is large enough, all Heavy Hitters will appear in the final set. 
Indeed, Space Saving has recently been identified as the most accurate and
efficient algorithm in practice for computing Heavy Hitters
\cite{survey}, and, as we later discuss, it also possesses strong
theoretical error guarantees \cite{2009}.

\newlength{\figwidthes}
\setlength{\figwidthes}{0.43\textwidth}
\begin{figure}[t]
\centering
\includegraphics[width=\figwidthes]{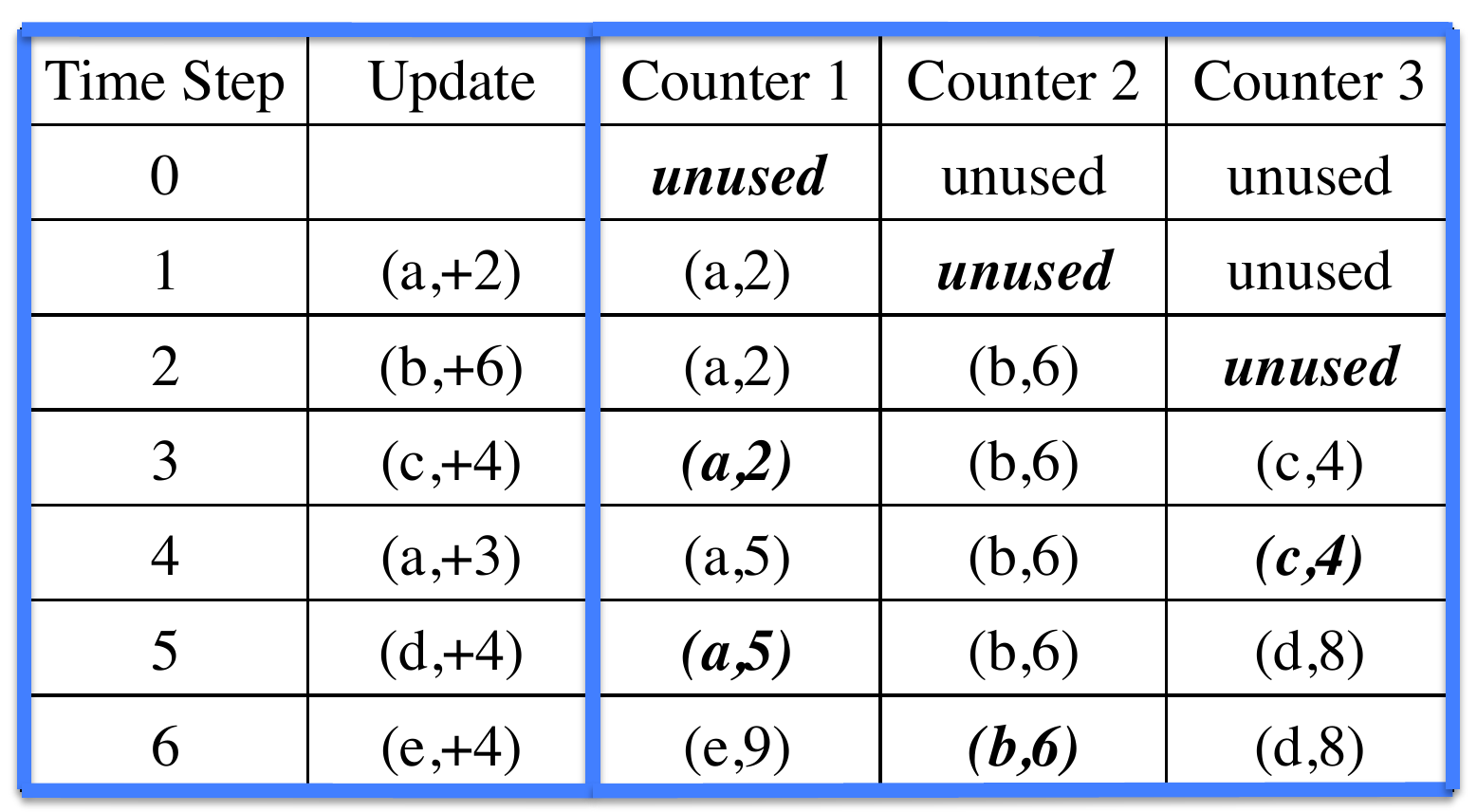}
\vspace{-2mm}
\caption{Sample execution of Space Saving with 3 counters. Each counter tracks an item (denoted by a letter), 
 and the estimated frequency of that item. The smallest counter is boldfaced and italicized.}
\label{fig:timep}
\hspace*{-6mm}
\end{figure}

\subsection{Related Work}

We require some notation to introduce prior related work and our
contributions; this notation is more formally defined in
Section~\ref{sec:definitions}.  In what follows, $N$ 
is the sum of all frequencies of items in the stream,
$\epsilon$ is an accuracy parameter so that all
outputs are within $\epsilon N$ of their true count, and $H$
represents the size of the hierarchy (specifically, the size of the
underlying lattice) the data belongs to.  Unitary updates refer to
systems where the count for an item increases by only 1 on each step,
or equivalently, where we just count item appearances.  

The one-dimensional HHH problem was first defined in \cite{seminal},
which also gave the first streaming algorithms for it.  Several
possible definitions and corresponding algorithms for the
multi-dimensional problem were introduced in \cite{conference,
  journal}. The definition we use here is the most natural, and was considered
  in several subsequent works \cite{pods, france}. In terms of
practical applications, multi-dimensional HHHs were used in \cite{offline,esv} to find patterns of traffic termed ``compressed traffic
clusters", in \cite{anomaly} for real-time anomaly detection, and in
\cite{DDos} for DDoS detection.

The Space Saving algorithm was used in \cite{china} in algorithms for
the one-dimensional HHH problem. Their algorithms require
$O(H^2 / \epsilon)$ space, while our algorithm
requires $O(H / \epsilon)$ space. Very recently, \cite{france} presented
an algorithm for the two-dimensional HHH problem, requiring $O(H^{3/2}/\epsilon)$
space.

Other recent work studies the HHH problem with a focus on developing algorithms well-suited
to commodity hardware such as ternary content-addressable memories (TCAMs) \cite{rexford}. 
Our algorithms are also well-suited to commodity hardware, as we describe in Section \ref{sec:extensions}.
The primary difference between the present work and \cite{rexford} is that the algorithms of \cite{rexford} 
reduce overhead by only updating rules periodically, rather than on a per-packet basis. This leads
to lightweight algorithms with \emph{no} provable accuracy guarantees. However, simulation results in \cite{rexford}
suggest these algorithms perform well in practice. In contrast, our algorithms possess very strong accuracy guarantees,
but likely result in more overhead than the approach of \cite{rexford}. Which approach is preferable may depend
on the setting and on the constraints of the data owner.
 
\subsection{Our Contributions}

In solving the Approximate HHH problem, there are three metrics that we seek to optimize: the time and space required to process each update and to output the list of approximate HHHs and their estimated frequencies and the quality of the output, in terms of the number of prefixes in the final output and the accuracy of the estimates. 
Our approach has several advantages over previous work. 

\begin{enumerate}
\item The worst-case space bound of our algorithm is $O(H/\epsilon)$. Notice this does not depend on the sum of the item frequencies, $N$, 
as $H$ depends only on the size of the underlying hierarchy and is independent of $N$.
This beats the worst-case space bound of $O\left( \frac{H}{\epsilon}\log{\epsilon N}\right)$ from \cite{conference} and \cite{journal}, the $O(H^2/\epsilon)$ bound for the one-dimensional algorithm of \cite{china}, and the $O(H^{3/2}/\epsilon)$ bound for the two-dimensional algorithm of \cite{france}. Additionally our algorithm provably requires $o(H/\epsilon)$ space under realistic assumptions on the frequency distribution of the stream. 

\item The worst-case time bound for our algorithm \emph{per} insertion
operation is $O(H \log{\frac{1}{\epsilon}})$ in the case of
arbitrary updates and $O(H)$ in the case of unitary updates.
Again this does not depend on $N$.  
Previous time bounds per insert were $O(H \log{\epsilon N})$ in \cite{seminal, conference, journal}.

\item  We obtain a refined analysis of error propagation to achieve better accuracy 
guarantees and provide non-trivial bounds on the number of HHHs output by our algorithm in one and two dimensions. These bounds were not provided for the algorithms in \cite{seminal, conference, journal}.

\item The space usage of our algorithm can be fixed \textit{a priori}, independent of the sum of frequencies $N$, as it only depends on the number of counters maintained by each instance of Space Saving, which we set at $\frac{1}{\epsilon}$ in the absence of assumptions about the data distribution. In contrast, the space usage of the algorithms of \cite{conference} and \cite{journal} depends on the input stream, and these algorithms dynamically add and prune counters over the course of execution, which can be infeasible in realistic settings. 

\item Our algorithm is conceptually simpler and substantially easier to implement than previous algorithms. We firmly believe \textit{programmer time} should be viewed as a resource similar to running time and space. We were able to use an off-the-shelf implementation of Space Saving, but this fact notwithstanding, we still spent roughly an order of magnitude less time implementing our algorithms, compared to those from \cite{conference, journal}. 

\item Our algorithms extend easily to more restricted settings. For example, we describe in Section \ref{sec:extensions} how to efficiently implement
our algorithms using TCAMs, how to parallelize them, how to apply them to distributed data streams, and how to handle sliding windows or
streams with deletions.
\end{enumerate}
\medskip
We present experimental results showing that for parameters of primary practical interest, our two-dimensional algorithm is superior to existing algorithms in terms of speed and accuracy, and competitive in terms of space, 
while our one-dimensional algorithm is 
also superior in terms of speed and accuracy for a more limited range of parameters. In short, we believe our algorithm offers a significantly better combination of simplicity and efficiency
than any existing algorithm.


\section{Notation, Definitions, and Setup}
\label{sec:definitions}

\subsection{Notation and Definitions}
As mentioned above, the theoretical framework developed in this section was described in \cite{journal}, and
considered in several subsequent works \cite{pods, france}.

In examples throughout this paper, we  consider the IP address
hierarchy at bytewise granularity: for example, the \textit{generalization} of
021.132.145.146 by one byte is 021.132.145.*, by two bytes is
021.132.*.*, by three bytes is 021.*.*.*, and by four bytes is
*.*.*.*. In two dimensions, we consider pairs of IP addresses,
corresponding to source and destination IPs. Each IP prefix that is
not fully general in either dimensions has two parents. For example,
the two parents of the IP pair (021.132.145.146, 123.122.121.120) are
(021.132.145.*, 123.122.121.120) and (021.132.145.146, 123.122.121.*).

In general, let the dimension of our data be $d$, and the height of the hierarchy in the $i$'th dimension be $h_i$. 
In the case of pairs of IP addresses, $d=2$ and $h_1=h_2=4$. Denote by $\text{par}(e, i)$ the generalization 
of element $e$ on dimension $i$; for example, if 
$$e=(021.132.145.*,123.122.121.120)$$then 
$\text{par}(e, 1)= 
(021.132.*.*, 123.122.121.120)$ and 
$\text{par}(e, 2)=(021.132.145.*, 123.122.121.*)$. 
Denote the 
generalization relation by $\prec$; for example, 
$$(021.132.145.*, 123.122.121.120) \prec (021.132.*.*, 
123.122.*.*).$$ 
Define $p \preceq q$ by $(p \prec q) \vee (p=q)$.  The generalization relation defines a 
lattice structure in the obvious manner. We overload our notation to define the sublattice of a 
\textit{set} of elements $P$ as $(e \preceq P) \Longleftrightarrow \exists p \in P$ such that $e 
\preceq P$. Let $H$ denote the total number of nodes in the lattice: $H = \prod_{i=1}^d (h_i+1)$. 

 We call an element \textit{fully specified}, if it is not the 
generalization of any other element, e.g. $021.132.145.163$ is fully specified. We 
call an element fully general in dimension $i$ if $\text{par}(e, i)$ does not exist. We refer to the 
unique element that is fully general in all dimensions as the \textit{root}. For ease of reference, we 
 label each element in the lattice with a vector of length $d$, whose $i$'th entry is at most 
$h_i$, to indicate which lattice node the element belongs to, with the vector corresponding to each 
fully specified element having $i$'th entry equal to $h_i$, and the vector corresponding to the root 
having all entries equal to 0. For example, the element $(021.132.145.*, 123.122.121.120)$ is assigned 
vector $(3, 4)$, and $(021.*.*.*, 123.122.121.*)$ is assigned vector $(1, 3)$. We define $\text{Level}(i)$ of 
the lattice to be the set of labels for which the sum of the entries in the label equals $i$. We  
overload terminology and refer to an element $p$ as a member of $\text{Level}(i)$ if the label assigned to $p$ 
is in $\text{Level}(i)$. Let $L=\sum_{i=1}^dh_i$ denote the deepest level in the hierarchy, that of the fully 
specified elements.

\begin{definition} \label{def:hhs} (Heavy Hitters) Given a multiset $S$ of size $N$ and a threshold $\phi$, a \emph{Heavy Hitter (HH)} is an element whose frequency in $S$ is no smaller than $\phi N$. Let $f(e)$ denote the frequency of each element $e$ in $S$. The set of heavy hitters is  $HH= \{e : f(e) \geq \phi N \}$.
\end{definition}

{From} here on, we assume we are given a multiset $S$ of (fully-specified) elements from a (possibly multidimensional) hierarchical domain of depth $L$, and a threshold $\phi$.

\begin{definition} \label{def:unconditioned} (Unconditioned count) Given a prefix $p$, define the \emph{unconditioned count} of $p$ as $f(p)=\sum_{e \in S \wedge e \preceq p} f(e)$. \end{definition} 

\noindent The exact HHHs are defined inductively as the 
set of prefixes whose \emph{conditioned count} 
exceeds $\phi N$, where the conditioned count is the sum of all 
descendant nodes that are neither HHHs themselves nor the descendant
of an HHH.  Formally:  

\begin{definition} \label{def:hhhs} (Exact HHHs) The set of exact \emph{Hierarchical Heavy Hitters} are defined inductively.

\begin{enumerate}
\item  $\mathcal{HHH}_L$, the hierarchical heavy hitters at level $L$, are the heavy hitters of $S$, that is the fully specified elements whose frequencies exceed $\phi N$. 

\item   
Given a prefix $p$ from $\text{Level}(l)$, $0 \leq l < L$, define
$\mathcal{HHH}^p_{l+1}$ to be the set $\{h \in \mathcal{HHH}_{l+1} \wedge h
\prec p\}$ i.e. $\mathcal{HHH}^p_{l+1}$ is the set of descendants of $p$
that have been identified as HHHs. Define the \emph{conditioned count}
of $p$ to be $F_p=\sum_{ (e \in S) \wedge (e \preceq p) \wedge (e
  \not\preceq \mathcal{HHH}^p_{l+1}) } f(e)$. The set
$\mathcal{HHH}_{l}$ is defined as $$\mathcal{HHH}_{l} = \mathcal{HHH}_{l+1} \cup \{p : (p
\in \text{Level}(l) \wedge (F_p \geq \phi N)\}.$$

\item   The set of exact Hierarchical Heavy Hitters $\mathcal{HHH}$ is defined as the set $\mathcal{HHH}_0$.

\end{enumerate} 
\end{definition}

Figure \ref{fig:twodexample} displays the exact HHHs for a two-dimensional hierarchy defined over an example stream. 

\begin{figure}[t]
\centering
\includegraphics[width=\figwidthes]{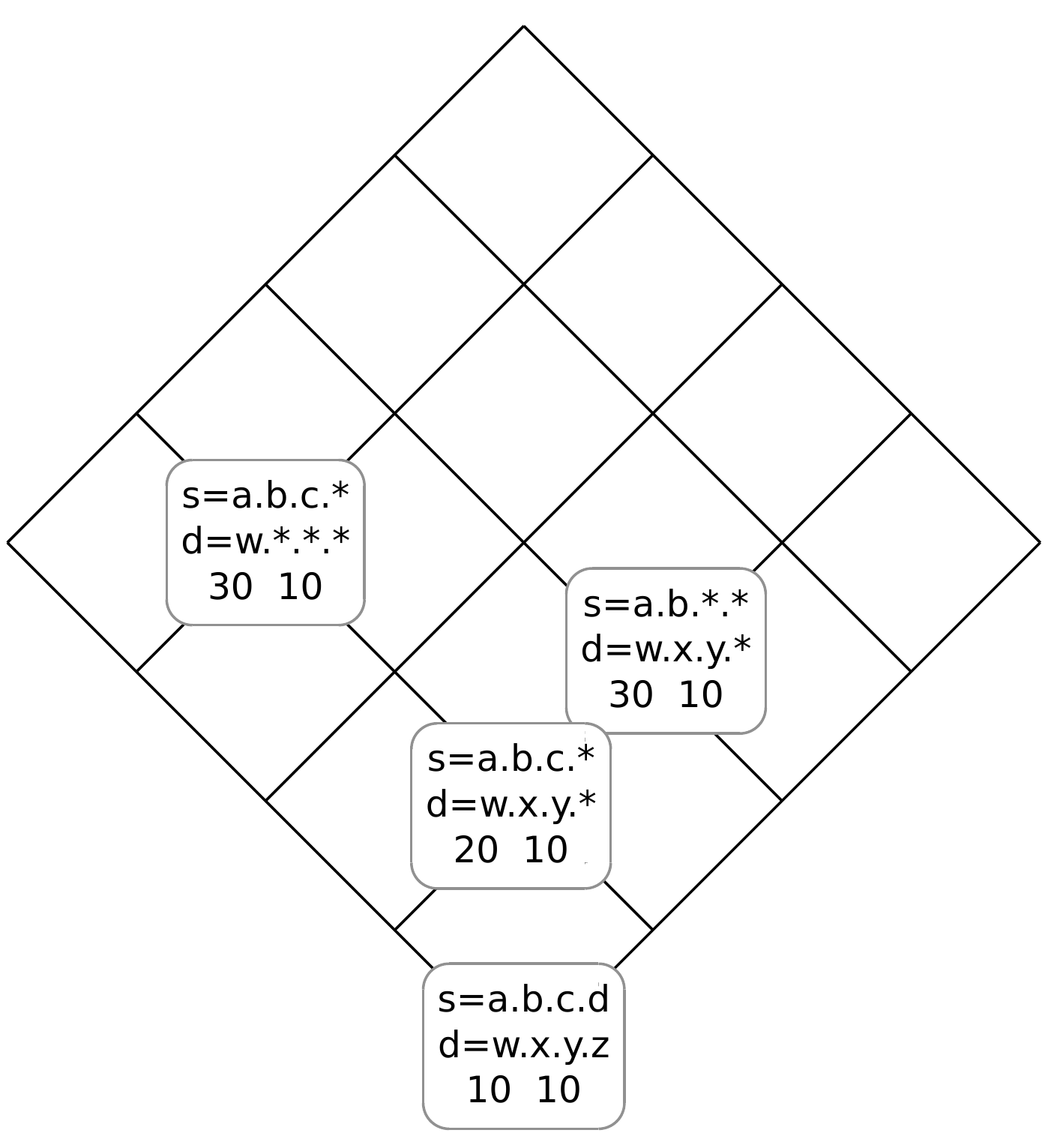}
\vspace{-2mm}
\caption{Example depicting exact HHHs for a two-dimensional stream of IP addresses at byte-wise granularity, using the
threshold $\phi N = 10$. The exact HHHs consist of \emph{ordered pairs} of source-destination IP-address prefixes (s denotes
source and d denotes destination). 
Unconditioned counts of each HHH are on the left, and conditioned counts for each HHH are on the right.
The stream consists of ten repetitions of the item
$(a.b.c.d,w.x.y.z)$, followed by one instance each of items $(a.b.c.i,w.x.y.i),
(a.b.i.d,w.x.y.i)$, and $(a.b.c.i,w.i.y.z)$ for all $i$ in the range 0 to
9. Here $a, b, c, d, w, x, y$, and $z$ represent some distinct integers
between 10 and 255.
}
\label{fig:twodexample}
\hspace*{-6mm}
\end{figure}

 
Finding the set of hierarchical heavy hitters and estimating their frequencies requires linear space to solve exactly, which is prohibitive. Indeed, even finding the set of heavy hitters requires linear space \cite{muthu}, and the hierarchical problem is even more general. For this reason, we study the \textit{approximate} HHH problem.

 
\begin{definition} \label{def:app} (Approximate HHHs) Given parameter $\epsilon$, the \emph{Approximate Hierarchical Heavy Hitters problem} with threshold $\phi$ is to output a set of items $P$ from the lattice, and lower and upper bounds $f_{\text{min}}(p)$ and $f_{\text{max}}(p)$, such that they satisfy two properties, as follows.

\begin{enumerate}
\item   \textit{Accuracy.} $f_{\text{min}}(p) \leq  f(p) \leq f_{\text{max}}(p)$, and $f_{\text{max}}(p)-f_{\text{min}}(p) \leq \epsilon N$ for all $p \in P$.

\item   \textit{Coverage.} For all prefixes $p$, define $P_p$ to be the set $\{q \in P: q \prec p\}$. Define the conditioned count of $p$ with respect to $P$ to be $F_p=\sum_{(e \in S) \wedge (e \preceq p) \wedge (e \not\preceq P_p)} f(e).$ We require for all prefixes $p \notin P$, $F_p < \phi N$. 
\end{enumerate}

\end{definition}

Intuitively, the Approximate HHH problem requires outputting a set $P$ such that 
no prefix with large conditioned count (with respect to $P$) is omitted, along with accurate 
estimates for the \emph{unconditioned} counts of prefixes in $P$. One might consider it natural
to require accurate estimates of the \emph{conditioned} counts of each $p \in P$ as well, but as
shown in \cite{pods}, $\Omega(1/\phi^{d+1})$ space would be necessary if
we required equally accurate estimates for the conditioned
counts, and this can be excessively large in practice.

\eat{
{\bf MM:  Some English here describing the definition in ``plain words'' might be helpful.
In particular, you say here in the definition you're using the conditioned count;  below
in related work you say this definition uses the unconditioned frequencies.  That's confusing.  In
fact, I'm moving that statement up here, where it belongs better.}
}

\subsection{Our Algorithm, Sketched}

Our algorithm utilizes the Space Saving algorithm, proposed by
Metwally et al. \cite{2005} as a subroutine, so we briefly describe
it and some of its relevant properties.  As mentioned, Space Saving
takes as input a stream of pairs $(i,c)$, where $i$ is an item and $c
> 0$ is a frequency increment for that item. It tracks a small subset $T$
of items from the stream with a counter for each $i \in T$.  If the
next item $i$ in the stream is in $T$, its counter is updated
appropriately.  Otherwise, the item with the smallest counter in $T$
is removed and replaced by $i$, and the counter for $i$ is set to the
counter value of the item replaced, plus $c$.
We now describe guarantees of Space Saving from \cite{2009}.

Let $N$
be the sum of all frequencies of items in the stream, let $m$ be the
number of counters maintained by Space Saving, and for any $j < m$,
let $N^{\text{res}(j)}$ denote the sum of all but the top $j$
frequencies. Berinde et al. \cite{2009} showed that for any $j <
m$, \begin{equation} \label{eq:ss} \forall i \left| f(i)-\hat{f}(i) \right| \leq
  \frac{N^{\text{res}(j)}}{m-j},\end{equation} where $\hat{f}(i)$ and $f(i)$ are the
  estimated and true
  frequencies of item $i$, respectively. By setting $j=0$,
this implies that $|f_i-\hat{f}_i|\leq \frac{N}{m}$, so only
$\frac{1}{\epsilon}$ counters are needed to ensure error at most
$\epsilon N$ in any estimated frequency. For frequency distributions
whose ``tails" fall off sufficiently quickly, Space Saving provably requires
$o(\frac{1}{\epsilon})$ space to ensure error at most
$\epsilon N$ (see \cite{2009} for more details).

Using a suitable min-heap based implementation of Space Saving,
insertions take $O(\log m)$ time, and lookups require $O(1)$ time
under arbitrary positive counter updates. When all updates are unitary
(of the form $c=1$), both insertions and lookups can be processed in
$O(1)$ time using the \emph{Stream Summary} data structure
\cite{2005}.

Our algorithm for HHH problems is conceptually simple: it keeps one
instance of a Heavy Hitter algorithm at each node in the lattice, and
for every update $e$ we compute all generalizations of $e$ and insert
each one separately into a different Heavy Hitter data structure. When
determining which prefixes to output as approximate HHHs, we start at
the bottom level of the lattice and work towards the top, using the
inclusion-exclusion principle to obtain estimates for the
\emph{conditioned} counts of each prefix. We output any prefix whose
estimated conditioned count exceeds the threshold $\phi N$.

We mention that the ideas underlying our algorithm have been implicit
in earlier work on HHHs, but have apparently been considered
impractical or otherwise inferior to more complicated
approaches. Notably, \cite{journal} briefly proposes an algorithm
similar to ours based on \emph{sketches}.
Their algorithm can handle deletions as well as insertions, but it requires more space and has
significantly less efficient output and insertion procedures. Significantly, this algorithm is only 
mentioned in \cite{journal} as an extension, and is not studied experimentally.
 An algorithm similar to
ours is also briefly described in \cite{pods} to show the asymptotic
tightness of a lower bound argument.  Interestingly, they clearly
state their algorithm is not meant to be practical. Finally, \cite{seminal}
describes a procedure similar to our one-dimensional algorithm, but concludes that
it is both slower and less space efficient than other algorithms. We therefore
consider one of our primary contributions to be the identification of
our approach as not only practical, but in fact superior in many
respects to previous more complicated approaches.

We chose the Space Saving algorithm \cite{2005} as our Heavy Hitter algorithm. In contrast, the algorithms of \cite{conference, journal} are conceptually based on the Lossy Counting Heavy Hitter algorithm \cite{lossycount}. A number of the advantages enjoyed by our algorithm can be traced directly to our choice of Space Saving over Lossy Counting, but not all. For example, the one-dimensional HHH algorithm of \cite{china} is also based on Space Saving, yet our algorithm has better space guarantees.

\section{One-Dimensional Hierarchies}
We now provide pseudocode for our algorithm in the one-dimensional case, which is much simpler than the case of arbitrary dimension.  As discussed, we use the Space Saving algorithm at each node of the
hierarchy, updating all appropriate nodes for each stream element, and then conservatively estimate conditioned counts to determine an appropriate output set.  

\vspace{-1mm}

\begin{codebox}
\Procname{$\proc{InitializeHHH}()$}
\li Initialize an instance $SS(n)$ of Space Saving with $\epsilon^{-1}$\\ counters at each node $n$ of the hierarchy.
\end{codebox}

\begin{codebox}
\li /*Line 4 tells the $n$'th instance of Space Saving\\to process $c$ insertions of prefix $p$*/
\Procname{$\proc{InsertHHH}$(element $e$, count $c$)}
\li \For all $p$ such that $e \preceq p$
\li 	\Do
			Let $n$ be the lattice node that $p$ belongs to
\li			UpdateSS($SS(n)$, $p$, $c$)
\end{codebox}
\vspace{-2mm}

\begin{codebox}
\Procname{$\proc{OutputHHH1D}$(threshold $\phi$)}
\li /* $\text{par}(e)$ is parent of $e$*/
\li Let $s_e=0$ for all $e$ 
\li /*$s_e$ conservatively estimates the difference\\
 between unconditioned and conditioned counts of $e$*/
\li \For each $e$ in postorder 
\li 	\Do
			$(f_{\text{min}}(e), f_{\text{max}}(e)) \gets$ GetEstimateSS($SS(n)$, $e$)
\li			\If $f_{\text{max}}(e) - s_e \geq \phi N$  
\li				\Do
				print($e$, $f_{\text{min}}(e)$, $f_{\text{max}}(e)$)
\li				$s_{\text{par}(e)} += f_{\text{min}}(e)$
\li			\Else
				$s_{\text{par}(e)} += s_e$				
\end{codebox}

Figure \ref{fig:onedalg} illustrates an execution of our one-dimensional algorithm on a stream of IP addresses at byte-wise granularity. 

\begin{figure}[!h]
\centering
\vspace{-2mm}
\includegraphics[width=\figwidthes]{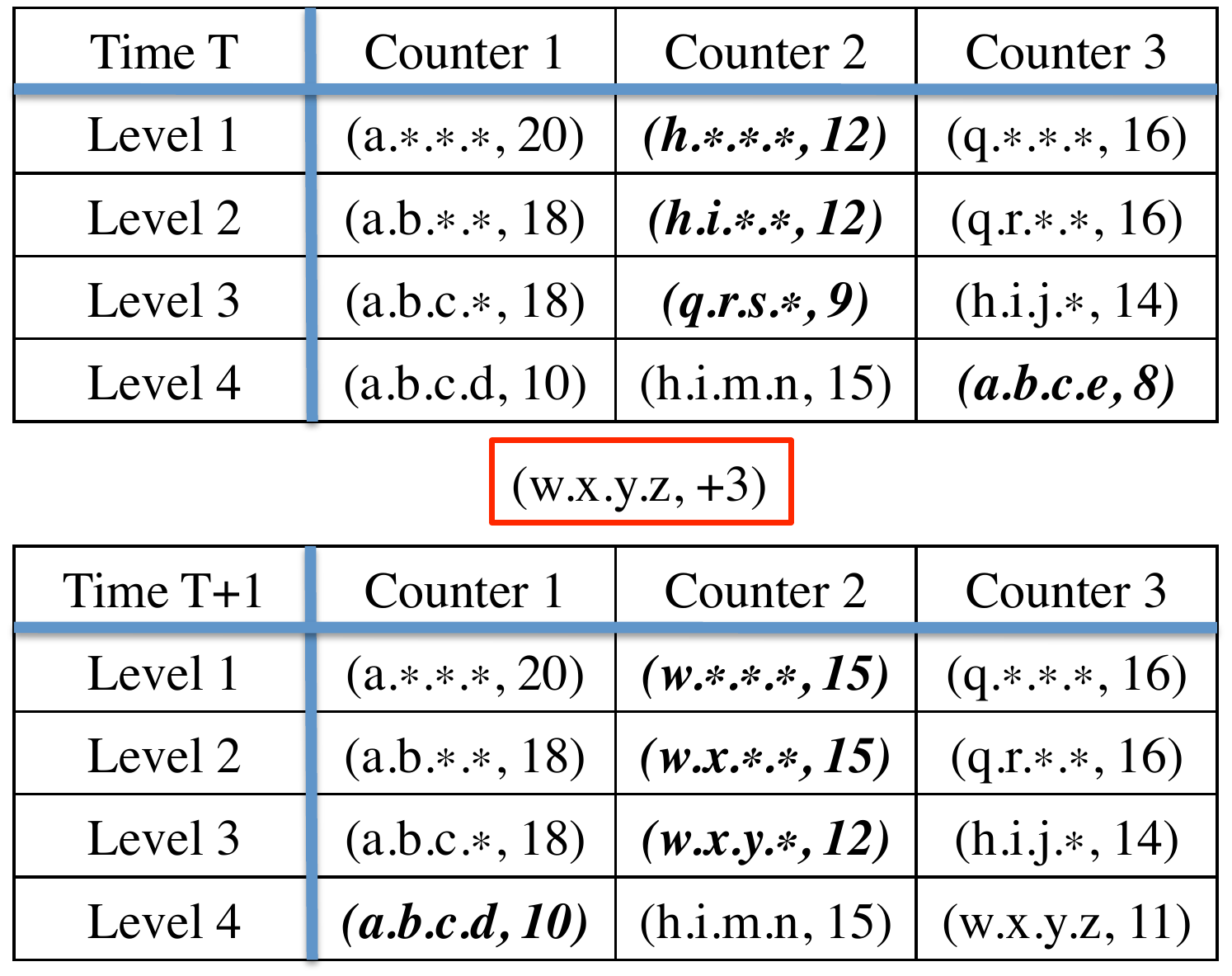}
\vspace{-2mm}
\caption{Example depicting our one-dimensional algorithm on a stream of IP addresses at byte-wise granularity, where each instance
of Space Saving maintains 3 counters. The top grid depicts the state at time $T$, and the bottom grid 
depicts the state at time $T+1$, after processing the update $(w.x.y.z, +3)$. The minimum counter for each instance of Space Saving
is boldfaced and italicized. If OutputHHH1D is run at time $T+1$ with threshold $\phi N = 12$, the approximate HHHs output would be
$h.i.m.n$, $w.x.y.*$, $h.i.j.*$, $a.b.c.*$, and $q.r.*.*$.
}
\label{fig:onedalg}
\hspace*{-6mm}
\vspace{-2mm}
\end{figure}

The following lemma is useful in proving that our one-dimensional algorithm satisfies various nice properties.  

\begin{lemma} \label{lem1} Define $H_p \subseteq P$ as the set $\{h: h\in P$, $h \prec p$, $\nexists h' \in P: h \prec h' \prec p\}$. Then in one dimension, $F_p=
f(p) - \sum_{h \in H_p} f(h)$. \end{lemma}
\begin{proof} By Definition \ref{def:app}, $$F_p=\sum_{(e \in S) \wedge (e \preceq p) \wedge (e \not\preceq P_p)} f(e) = f(p)-\sum_{(e \in S) \wedge (e \preceq P_p)}f(e).$$ Since the hierarchy is one-dimensional, for each
$e \in S$ such that $e \preceq P_p$, there is exactly one $h \in H_p$ such that $e \preceq h$ (otherwise, there would be $h \neq h'$ in $H_p$ such that $h \prec h'$). Thus, \begin{align*} f(p)-\sum_{(e \in S) \wedge (e \preceq P_p)}f(e) &=f(p) - \sum_{h \in H_p} \sum_{(e \in S) \wedge (e \preceq h)} f(h)\\
&= f(p) - \sum_{h \in H_p} f(h). \end{align*} \smartqed \end{proof}

\begin{theorem} \label{thm1} Using $O(\frac{H}{\epsilon})$ space, our one-dimensional algorithm satisfies the Accuracy and Coverage requirements of Definition \ref{def:app}. 
\end{theorem}
\begin{proof}			
By Equation \ref{eq:ss}, each instance of Space Saving requires $\frac{1}{\epsilon}$ counters, corresponding to O$(\frac{1}{\epsilon})$ space, in order to estimate the unconditioned frequency of each item assigned to it within additive error $\epsilon N$. Consequently, the Accuracy requirement is satisfied using $O(\frac{H}{\epsilon})$ space in total.

To prove coverage, we first show by induction that $s_p = \sum_{h \in H_p} f_{\text{min}}(h)$. This is true at level $L$ because in this case $s_p=0$ and $H_p$ is empty. Suppose the claim is true for all prefixes at level $k$. Then for  $p$ at level $k-1$,

\begin{align*} s_p&=\sum_{q \in \text{child}(p) \wedge q \in P} f_{\text{min}}(q) + \sum_{q \in \text{child}(p) \wedge q \notin P} s_q\\
&= \sum_{q \in \text{child}(p) \wedge q \in P} f_{\text{min}}(q) + \sum_{q \in \text{child}(p) \wedge q \notin P} \sum_{h \in H_q} f_{\text{min}}(h)\\
&=  \sum_{h \in H_p} f_{\text{min}}(h),\end{align*}
where the first equality holds by inspection of Lines 5-9 of the output procedure, and the second equality holds by the inductive hypothesis. This completes the induction.

By Lemma \ref{lem1}, $F_p= f(p) - \sum_{h \in H_p}f(h)$
$$ \leq f_{\text{max}}(p) - \sum_{h \in H_p}f_{\text{min}} (h)= f_{\text{max}}(p)-s_p,$$ where the inequality holds by the Accuracy guarantees. Coverage follows, since our algorithm is conservative. That is, if item $p$ is not output, then from Line 6 of the output procedure we have $f_{\text{max}}(p) - s_p \leq \phi N$, and we've shown  $F_p \leq f_{\text{max}}(p) - s_p$.
\end{proof}

We remark that under realistic assumptions on the data distribution, our algorithm satisfies the Accuracy and Coverage requirements using 
space $o(\frac{H}{\epsilon})$. Specifically, \cite[Theorem 8]{2009} shows that, if the tail of the frequency distribution (i.e. the quantity $N^{\text{res}(k)}$ for a certain value of $k$) is bounded by that of the Zipfian distribution with parameter $\alpha$, then Space Saving requires space $O(\epsilon^{-\frac{1}{\alpha}})$ to estimate all frequencies within error $\epsilon N$. Notice that if the frequency distribution of the stream itself satisfies this ``bounded-tail" condition, then the frequency distributions at higher levels of the hierarchy do as well.
Hence our algorithm requires only space $O(H \epsilon^{-\frac{1}{\alpha}})$ if the tail of the stream is bounded by that of a Zipfian distribution with 
parameter $\alpha$.

\begin{theorem} \label{thm2} Our one-dimensional algorithm performs each update operation in time $O(H\log{\frac{1}{\epsilon}})$ in the case of arbitrary updates, and $O(H)$ time in the case of unitary updates. Each output operation takes time $O(\frac{H}{\epsilon})$.\end{theorem}
\begin{proof} 
 The time bound on insertions is trivial, as an insertion operation requires updating $H$ instances of Space Saving. Each update of Space Saving using a min-heap based implementation for arbitrary updates requires time $O(\log{m})$, where $m=\frac{1}{\epsilon}$ is the number of counters maintained by each instance of Space Saving. For unitary updates, each insertion to Space Saving can be processed in $O(1)$ time using the \emph{Stream Summary} data structure \cite{2005}.

To obtain the time bound on output operations, notice that although the pseudocode for procedure OutputHHH1D indicates that we iterate through every possible prefix $e$, we actually need only iterate over those $e$ tracked by the instance of Space Saving corresponding to $e$'s label.  We may restrict our search to these $e$ because, for any prefix $e$ not tracked by the corresponding Space Saving instance, $f_{\text{max}}(e) \leq \epsilon N < \phi N$, so $e$ cannot be an approximate HHH. There are at most $\frac{H}{\epsilon}$ such $e$'s because each of the $H$ instances of Space Saving maintains only $\frac{1}{\epsilon}$ counters, and for each $e$, the GetEstimateSS call in line 5 and all operations in lines 6-9 require $O(1)$ time. The time bound follows.
\end{proof}

For all prefixes $p$ in the lattice, define the estimated conditioned count of $p$ to be $F'_p := f_{\text{max}}(p)-s_p$. By performing a refined analysis of error propagation, we can bound the number of HHHs output by our one-dimensional algorithm, and use this result to provide Accuracy guarantees on the estimated conditioned counts.

\begin{theorem} \label{thm:conditioned} Let $\epsilon < \frac{\phi}{2}$. The total number of approximate HHHs output by our one-dimensional algorithm is at most $\frac{1}{\phi-2\epsilon}$.  Moreover, the maximum error in the approximate conditioned counts, $F'_p-F_p$, is at most $\frac{1}{\phi-2\epsilon}\epsilon N$. \end{theorem}


\begin{proof} 
We first sketch why not too many approximate HHHs are output. A prefix $p$ is output if and only if $F'_p > \phi N$, and $F'_p \geq F_p$. The key observation is that for each approximate HHH $h \in P$ output by our algorithm, $h$ ``contributes" error at most $\epsilon N$ to the estimated conditioned count $F'_p$ of \textit{at most one} ancestor $p \in P$ of $h$. Therefore, the total error in the approximate conditioned counts of the output set $P$ is small. Consequently, the sum of the \textit{true} conditioned  counts $F_p$ of all $p \in P$ is very close to $\phi N |P|$, implying that $|P|$ cannot be much larger than $\frac{N}{\phi}$ since the stream has length $N$.

We make this argument precise. We showed in proving Theorem \ref{thm1} that  for all $p$, $s_p = \sum_{h \in H_p} f_{\text{min}}(h)$, so \begin{equation} \label{eqn:one} F'_p = f_{\text{max}}(p)-s_p = f_{\text{max}}(p) -\sum_{h \in H_p} f_{\text{min}}(h).\end{equation} 

Combining Lemma \ref{lem1} and Equation \ref{eqn:one}, we see that 
$$F'_p-F_p = \big(f_{\text{max}}(p)-\sum_{h \in H_p} f_{\text{min}}(h) \big) - \big(f(p)-\sum_{h \in H_p} f(h) \big)= $$
 $$  \big(f_{\text{max}}(p) - f(p) \big) + \big(\sum_{h \in H_p}  f(h) - f_{\text{min}}(h) \big) .$$
 
To show that the sum of the \textit{true} conditioned counts $F_p$ of all $p \in P$ is very close to $\phi N |P|$, we use
$$\sum_{p \in P} F_p = \sum_{p \in P} F'_p - \sum_{p \in P} (F'_p - F_p) $$
$$\geq |P| \phi N - \sum_{p \in P} \big(f_{\text{max}}(p) - f(p) \big) - \sum_{p \in P} \big(\sum_{h \in H_p}  f(h) - f_{\text{min}}(h) \big).$$

By the Accuracy guarantees, $\sum_{p \in P} \big(f_{\text{max}}(p) - f(p) \big)$ is at most $|P| \epsilon N$.  To bound $\sum_{p \in P} \big(\sum_{h \in H_p}  f(h) - f_{\text{min}}(h) \big)$, we observe that for any item $h \in P$, $h \in H_p$ for at most one ancestor $p \in P$ (because in one dimension, if $h \prec p$ and $h \prec p'$ for distinct $p, p' \in P$, then either $p \prec p'$ or $p' \prec p$, contradicting the fact that $h \in H_p$ and $h \in H_{p'}$). Combining this fact with the Accuracy guarantees, we again obtain an upper bound of $|P| \epsilon N$. 
In summary, we have shown that $$\sum_{p \in P} F_p \geq |P| \phi N - 2 \epsilon |P| N = |P| (\phi - 2 \epsilon) N.$$ Since the total length of the stream is $N$, and in one dimension each fully specified item contributes its count to $F_p$ for at most one $p$,  it follows that $\sum_{p \in P} F_p \leq N$ and hence $|P| \leq \frac{1}{\phi-2 \epsilon}$ as claimed.


Lastly, we bound the maximum error $F'_p-F_p$ in any estimated conditioned  count. We showed that $$F'_p-F_p = \big(f_{\text{max}}(p) - f(p) \big) + \big(\sum_{h \in H_p}  f(h) - f_{\text{min}}(h) \big),$$ which, by the Accuracy guarantees, is at most $\epsilon N + |H_p| \epsilon N \leq \epsilon N + (|P| - 1) \epsilon N \leq \frac{\epsilon}{\phi-2 \epsilon} N$, as claimed.
\end{proof}

The upper bound on output size provided in Theorem \ref{thm:conditioned} is very nearly tight, as there may be $\frac{1}{\phi}$ exact heavy hitters. For example, with realistic values of $\phi=.01$ and $\epsilon=.001$, Theorem 3 yields an upper bound of 102.

\section{Two-Dimensional Hierarchies}

In moving from one to multiple dimensions, only the output procedure must change. In one dimension, discounting items that were already output as HHHs was simple. There was no double-counting involved, since no two children of an item $p$ had common descendants. To deal with the double-counting, we use the principle of inclusion-exclusion in a manner similar to \cite{journal} and \cite{conference}. 

At a high level, our two-dimensional output procedure works as follows. As before, we start at the bottom of the lattice, and compute HHHs one level at a time. For any node $p$, we have to estimate the conditioned count for $p$ by discounting the counts of items that are already output as HHHs. However, Lemma \ref{lem1} no longer holds: it is not necessarily true that $F_p =  f_{\text{max}}(p)-\sum_{q \in H_p} f(q)$ in two or more dimensions, because for fully specified items that have two or more ancestors $H_p$, we have subtracted their count multiple times. Our algorithm compensates by adding these counts back into the sum. 

Before formally presenting our two-dimensional algorithm, we need the following theorem. Let $\text{glb}(h, h')$ denote the greatest lower bound of $h$ and $h'$, that is, the unique common descendant $q$ of $h$ and $h'$ satisfying $\forall p: (q \preceq p) \wedge (p \preceq h) \wedge (p \preceq h') \Longrightarrow p = q$. In the case where $h$ and $h'$ have \textit{no} common descendants, the we treat $\text{glb}(h, h')$ as the ``trivial item" which has count 0. 

\begin{theorem} \label{thm:multid} In two dimensions, let $T_{p}$ be the set of all $q$  expressible as the greatest lower bound of two distinct elements of $H_p$, but not of $3$ or more distinct elements in $H_p$. 
Then $$F_p=f(p)-\sum_{q \in H_p} f(q) + \sum_{q\in T_{p} } f(q).$$ \end{theorem}
The proof appears in Appendix~\ref{app:proof}.
 

Below, we give pseudocode for our two-dimensional output procedure. 
We compute estimated conditioned counts
$F'_p=f_{\text{max}}(p)-\sum_{h_1 \in H_p} f_{\text{min}}(h_1) + \sum_{q \in T_p} f_{\text{max}}(q).$ 
As in the one-dimensional case, the Accuracy guarantees of the algorithm follow
immediately from those of Space Saving.
Coverage requirements are satisfied
by combining Theorem \ref{thm:multid} with the Accuracy guarantees.

Our two-dimensional algorithm performs each insert operation in $O(H \log{\frac{1}{\epsilon}})$ time under arbitrary updates, and $O(H)$ time under unitary updates, just as in the one-dimensional case. Although the output operation is considerably more expensive in the multi-dimensional case, experimental results indicate that this operation is not prohibitive in practice (see Section \ref{sec:experiments}).

  \begin{codebox}
\Procname{$\proc{OutputHHH2D}$(threshold $\phi$)}
\li  $P=\emptyset$
\li \For \textit{level l=L} \Downto 0
 	\Do
\li			\For each item $p$ at level $l$
				\Do
\li					Let $n$ be the lattice node that $p$ belongs to
\li					$(f_{\text{min}}(p), f_{\text{max}}(p))$\!$\gets$\!GetEstimateSS($SS(n)$, $p$)
\li					$F'_p = f_{\text{max}}(p)$
\li					$H_p = \{h\!\in\!P$ such that $\nexists h' \in\!P: h \prec h' \prec p\}$
\li					\For each $h \in H_p$ 
						\Do
\li							$F'_p = F'_p - f_{\text{min}}(h)$
						\End
\li					\For each pair of distinct elements $h, h'$ in $H_p$
						\Do
\li							$q=\text{glb}(h, h')$
\li                      \If $\nexists h_3 \neq h, h'$ in $H_p$ s.t. $q \preceq h_3$
						  \Do
\li							$F'_p = F_p' + f_{\text{max}}(q)$	
						  \End					
						\End
\li					\If $F'_p \geq \phi N$
						\Do
\li							$P = P \cup \{p\}$
\li							print($p$, $f_{\text{min}}(p)$, $f_{\text{max}}(p)$)
		
\end{codebox}

Using Theorem \ref{thm:multid}, we obtain a non-trivial upper bound on the number of HHHs output by our two-dimensional algorithm. The proof is in Appendix \ref{app:proof}.

\begin{theorem} \label{thm:multidbound} Let $A = 1+\min(h_1, h_2)$, where $h_i$ is the depth of dimension $i$ of the lattice. For small enough $\epsilon$, the number of approximate HHHs output by our two-dimensional algorithm is at most 
$$\frac{2}{A\epsilon}\left(\phi - (1 + A) \epsilon - \sqrt{(\phi - (1 + A) \epsilon)^2 - A^2\epsilon}\right).$$
\end{theorem}

The error guarantee obtained from Theorem \ref{thm:multidbound} appears messy, but yields useful bounds in many realistic settings. For example, for IP 
addresses at byte-wise granularity, $A = 5$. Plugging in $\phi=.1$, $\epsilon = 10^{-4}$ yields $|P| \leq 53$, which is very close to the maximum number of exact 
HHHs: $A/\phi = 50$.  
As further examples, setting $\phi=.05$ and $\epsilon=10^{-5}$ yields a bound of $|P| \leq 102$, and setting $\phi=.01$ and $\epsilon=10^{-6}$ yields a bound of $|
P| \leq 536$, both of which are reasonably close to $\frac{A}{\phi}$.
Of course, the bound of Theorem \ref{thm:multidbound} should not be viewed as tight in practice, but rather as a worst-case guarantee on output size.


\medskip
\noindent \textbf{Higher Dimensions.}
In higher dimensions, we can again keep one instance of Space Saving at each node of the hierarchy to compute estimates $f_{\text{min}}(p)$ and $f_{\text{max}}(p)$ of the unconditioned count of each prefix $p$. We need only modify the Output procedure to conservatively estimate the \emph{conditioned} count of each prefix. 

We can show that the natural generalization of Theorem \ref{thm:multid} does not hold in three dimensions.  However, we can compute estimated conditioned sublattice counts $F'_p$ as 
$$F'_p = f(p)-\sum_{h \in H_p} f_{\text{min}}(h) + \sum_{(h\in H_p, h' \in
  H_p) \wedge q = \text{glb}(h, h')} f_{\text{max}}(q).$$  Inclusion-exclusion
implies that, in any dimension, $F_p \leq F'_p$, and hence by
outputting $p$ if $F'_p \geq \phi N$ we can satisfy Coverage.

\section{Extensions}
\label{sec:extensions}
Our algorithms are easily adopted to distributed or parallel settings, and can be efficiently implemented in commodity hardware such as ternary content addressable memories.

\medskip
\noindent \textbf{Distributed Implementation.}
In many practical scenarios a data stream is distributed across
several locations rather than localized at a central node (see, e.g.,
\cite{distributed1, distributed2}). For example, multiple sensors
may be distributed across a network. We extend our algorithms
to this setting.

Multiple independent instances of Space Saving can be merged
to obtain a single summary of the concatenation of the distributed data streams
with only a constant factor loss in accuracy, as shown in \cite{2009}.
We use this form of their result:
\begin{theorem} (\cite[Theorem 11]{2009}, simplified statement): Given summaries of $k$ distributed data streams produced by $k$ instances of Space Saving each with $\frac{1}{\epsilon}$ counters, a summary of the concatenated stream can be obtained such that the error in any estimated frequency is at most $3\epsilon N$, where $N$ is the length of the concatenated stream. \end{theorem}

To handle $k$ distributed data streams, we may simply run one instance of our algorithm independently on each stream (with $\frac{3}{\epsilon}$ counters each), and afterward, for each node in the lattice, merge all $k$ corresponding instances of Space Saving into a single instance.
After the merge, we have a single instance of Space Saving for each node in the lattice that has essentially the same error guarantees (up to a small constant factor) as a centralized implementation. Our output procedure is exactly as in the centralized implementation.


\medskip
\noindent \textbf{Parallel Implementation.}
In all of our algorithms, the update operation involves updating a number of independent Space Saving instances. It is therefore trivial to parallelize this algorithm. 
We have parallelized this algorithm using OpenMP.  Our limited experiments show essentially linear speedup, up to the point where we reach the limitation of the shared memory constraint.

\eat{However, the real time does not decrease quite linearly, indicating that there is a memory access bottleneck. {\bf I think I'd be in favor of cutting the real-time plot for space reasons, but if we leave it in, we need to explain this last sentence more. I thought if the CPU is waiting for memory access, it counts against both CPU time and real time? So I don't think a memory bottleneck is the explanation. Maybe it could be caused by context switching? MM do you have thoughts on what might be causing this, and/or agree with just cutting the ``real-time'' graph, especially in light of the fact that wumpus is not a ``controlled system''? --JT}}

\medskip
\noindent \textbf{TCAM Implementation.}
Recently, there has been an effort to develop network algorithms that
utilize Ternary Content Addressable Memories, or
\emph{TCAMs}, to process streaming queries faster.  TCAMs are
specialized, widely deployed hardware that support constant-time
queries for a bit vector within a database of ternary vectors, where
every bit position represents $0$, $1$ or $*$. The $*$ is a wild card
bit matching either a $0$ or a $1$. In any query, if there is one
or more match, the address of the highest-priority match is returned.
Previous work describes a TCAM-based implementation of Space Saving
for unitary updates, and shows experimentally that 
it is several times faster than software solutions
\cite{tcam}.

Since our algorithms require the maintenance of $H$ independent
instances of Space Saving, it is easy to see that our algorithms can
be implemented given access to $H$ separate TCAMs, each
requiring just a few KBs of memory. With more effort, we can devise
implementations of our algorithms that use just
a single commodity TCAM. Commodity TCAMs can store hundreds of thousands or
millions of data entries \cite{tcam}, and therefore a single TCAM can
store tens of instances of Space Saving even when $\epsilon=.0001$.

Our simplest TCAM-based implementation takes advantage of the fact
that TCAMs have \emph{extra bits}. A typical TCAM has a width of 144
symbols allotted for each entry in the database, and this typically
leaves several dozen unused symbols for each entry. The implementation
of Space Saving of \cite{tcam} uses extra bits to store frequencies,
but we can use additional unused bits to identify the instance of
Space Saving associated with each item in the database.

For illustration, consider the one-dimensional byte-wise IP hierarchy. We associate two extra bits with each entry in the database: 00 will correspond to the top-most level of the hierarchy, 01 to the second level, 10 to the third, and 11 to the fourth. Then we treat each IP address $a.b.c.d$ as four separate searches:
$a.b.c.d$.00, $a.b.c.*$.01, $a.b.*.*$.10, and $a.*.*.*$.11, thereby updating each ancestor
of $a.b.c.d$ in turn. The TCAM needs to store the smallest counter for each
of the four Space Saving instances, and otherwise the TCAM-based implementation from
\cite{tcam} is easily modified to handle multiple 
instances of Space Saving on a single TCAM. 

Alternatively, we could compute approximate unconditioned counts by keeping a \emph{single} instance of Space Saving with (item, mask) pairs as keys, rather than $H$ separate instances of Space Saving. It is clear that this approach still satisfies the Accuracy guarantees for each prefix, and has the advantage of only having to store the smallest counter for a single instance of Space Saving.

\medskip
\noindent \textbf{Sliding Windows and Streams with Deletions.}
Our algorithms as described only work for insert-only
streams, due to our choice of Space Saving
as our heavy hitter algorithm. However, the accuracy and coverage
guarantees of our HHH algorithms still hold even if we replace Space
Saving with other heavy hitter algorithms.  This is
because our proofs of accuracy and coverage applied the
inclusion-exclusion principle to express conditioned counts in terms
of unconditioned counts, and then used the fact that our heavy hitter
algorithm provides accurate estimates on the unconditioned counts;
this analysis is independent of the heavy hitter algorithm used.
Hence we can extend our results to additional scenarios by using other
algorithms.

For example, it may be desirable to compute HHHs over only a sliding window of the last $n$ items seen in the stream.  
\cite{sliding} presents a deterministic algorithm for computing $\epsilon$-approximate heavy hitters over sliding windows using $O(1/\epsilon)$ space. 
Thus, by replacing Space Saving with this algorithm, we obtain an algorithm that computes approximate 
HHHs over sliding
windows using space $O(H/\epsilon)$, which asymptotically matches the space usage of our algorithm. However,
it appears this algorithm is markedly slower and less space-efficient in practice. 
 
Similarly, many \emph{sketch-based} heavy hitter algorithms such as that of \cite{count-min} can compute $\epsilon$-approximate heavy hitters,
even in the presence of deletions, using space $O(\frac{1}{\epsilon} \log N)$. By replacing Space Saving with such a sketch-based algorithm, we obtain
a HHH algorithm using space $O(\frac{H}{\epsilon} \log N)$ that can handle streams with deletions.  (As noted previously, this variation was
mentioned in \cite{journal}.)
\vspace{-3mm}

\section{Experimental Results}
\label{sec:experiments}

We have implemented two versions of our algorithm in \texttt{C} and tested it using GCC version 4.1.2 on a host with four single-core 64-bit AMD Opteron 850 processors each running at 2.4GHz with a 1MB cache and 8GB of shared memory. The first version -- termed \texttt{hhh} below -- uses a heap-based implementation of Space Saving that can handle arbitrary updates, while the second version -- termed \texttt{unitary} below -- uses the \emph{Stream Summary} data structure and can only handle unitary updates. Both versions use an off-the-shelf implementation from \cite{code} for Space Saving; further optimizations, as well as different tradeoffs between time and space, may be possible by modifying the off-the-shelf implementation. We have used a real packet trace from \texttt{www.caida.org} \cite{caida} in all experiments below.  (We have tried other traces to confirm that these results are demonstrative. Note that all of our graphs are in color and may not display well in grayscale.)  Throughout our experiments, all algorithms define the frequency of an IP address or an IP address pair to be the number of packets associated with that item, as opposed to the number of bytes of raw data. This ensures that all algorithms (including \texttt{unitary}) process exactly the same updates. Consequently, the stream length $N$ in all of our experiments refers to the number of packets in the stream (i.e. the prefix of the packet trace \cite{caida} that we used). 

We tested our algorithms at both byte-wise and bit-wise granularities in one and two dimensions. Bit-wise hierarchies are more expensive to handle, as $H$, the number of nodes in the lattice structure implied by the hierarchy, becomes much larger. However, it may be useful to track approximate HHHs at bit-wise granularity in many realistic situations. For example, a single entity might control a subnet of IP addresses spanning just a few bits rather than an entire byte. However, we observed similar (relative) behavior 
between all algorithms at both bit-wise and byte-wise granularity, and thus we display results only for byte-wise hierarchies for succinctness. 

For comparison we also implemented the full and partial ancestry algorithms from \cite{journal}, labeled \texttt{full} and \texttt{partial} 
respectively. 
We compare the algorithms' performance in several respects: time and memory usage, the size of the output set, and the accuracy of the unconditioned count estimates. Our algorithm performs at least as well as the other two in terms of output size and accuracy. Except for extremely small values of $\epsilon$ (less than about .0001), which correspond to extremely high accuracy guarantees, our two-dimensional algorithm is also significantly faster (more than three times faster for some parameter settings of high practical interest).  Our one-dimensional algorithm is also faster than its competitors for values of $\epsilon$ greater than about .01, and competitive across all values of $\epsilon$. Our algorithm uses slightly more memory than its competitors. Below, we discuss each aspect separately.

In summary, our one-dimensional algorithm is competitive in practice
with existing solutions, and possesses other desirable properties that
existing solutions lack, such as improved simplicity and ease of
implementation, improved worst-case guarantees, and the ability to
preallocate space even without knowledge of the stream length. Our
two-dimensional algorithm possesses all of the same desirable
properties, and is also significantly faster than existing solutions
for parameter values of primary practical interest. The primary
disadvantage of our algorithms is slightly increased space usage.

All of our implementations are available online at \cite{ourcode}.

\newlength{\figwidth}
\setlength{\figwidth}{0.30\textwidth}
\begin{figure*}
\centering
\subfloat[Maximum memory usage in one dimension over all stream lengths $N$. For \texttt{hhh} and \texttt{unitary}, space usage does \emph{not} depend on $N$; space usage only varied with $N$ for \texttt{partial} and \texttt{full}.]
{
\includegraphics[width=\figwidth]{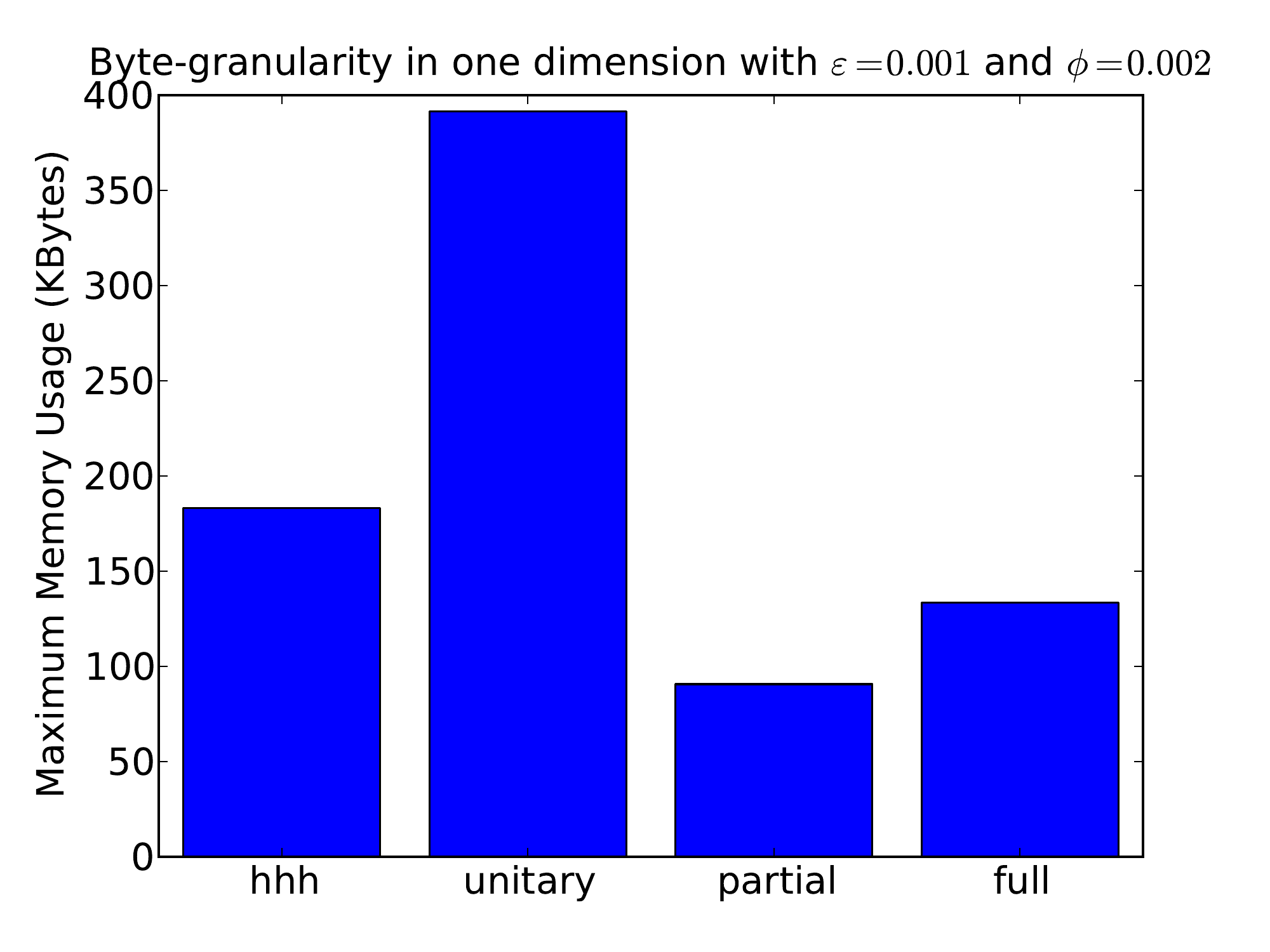}
\label{figure:memory_1000-500_1}
}%
\hspace{0.08in}
\subfloat[Maximum memory usage in two dimensions over all stream lengths $N$. For \texttt{hhh} and \texttt{unitary}, space usage does \emph{not} depend on $N$; space usage only varied with $N$ for \texttt{partial} and \texttt{full}. ]
{
\includegraphics[width=\figwidth]{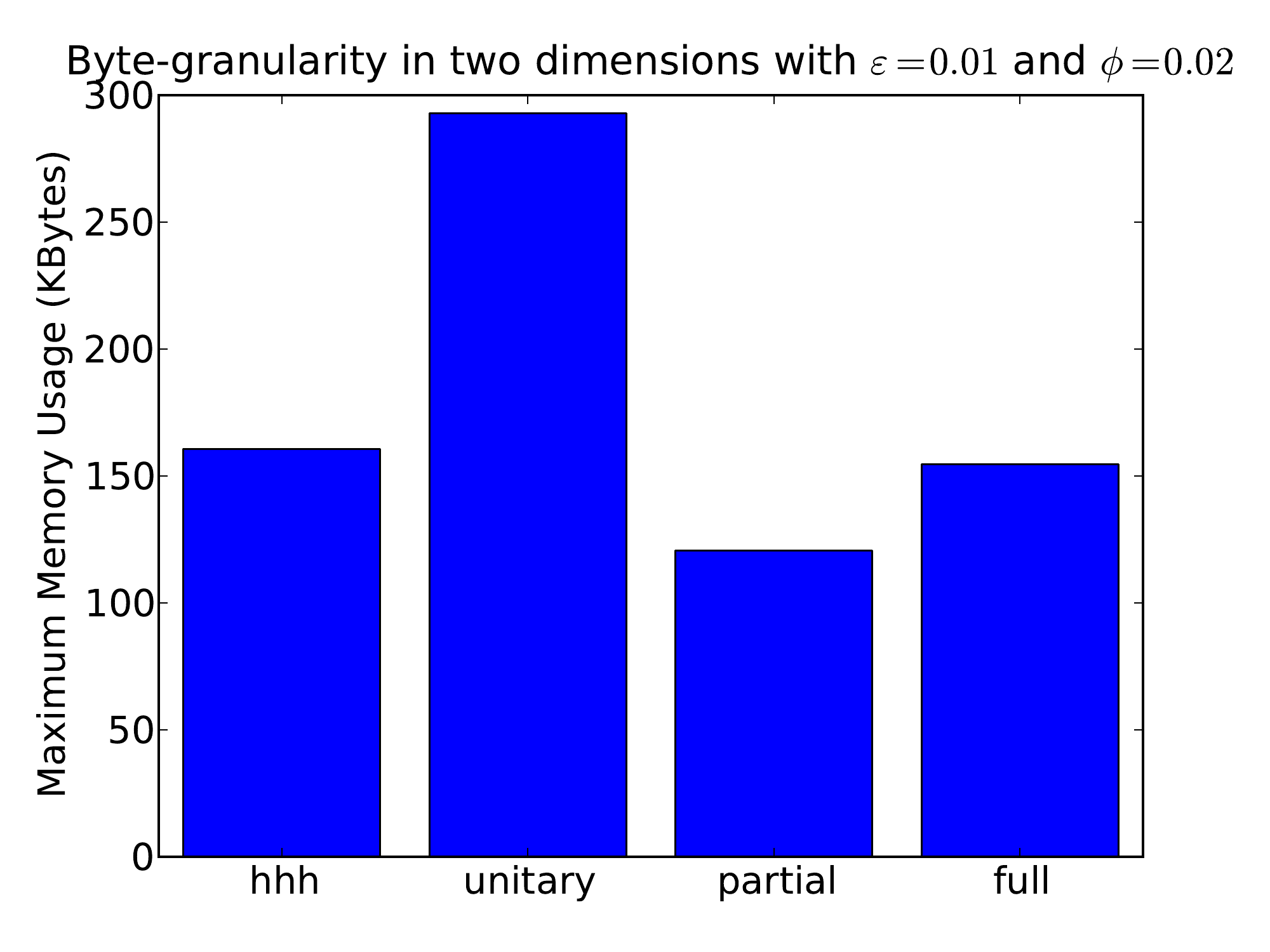}
\label{figure:memory_100-50_2}
}%
\hspace{0.08in}
\subfloat[Memory usage in one dimension for fixed stream length.]
{
\includegraphics[width=\figwidth]{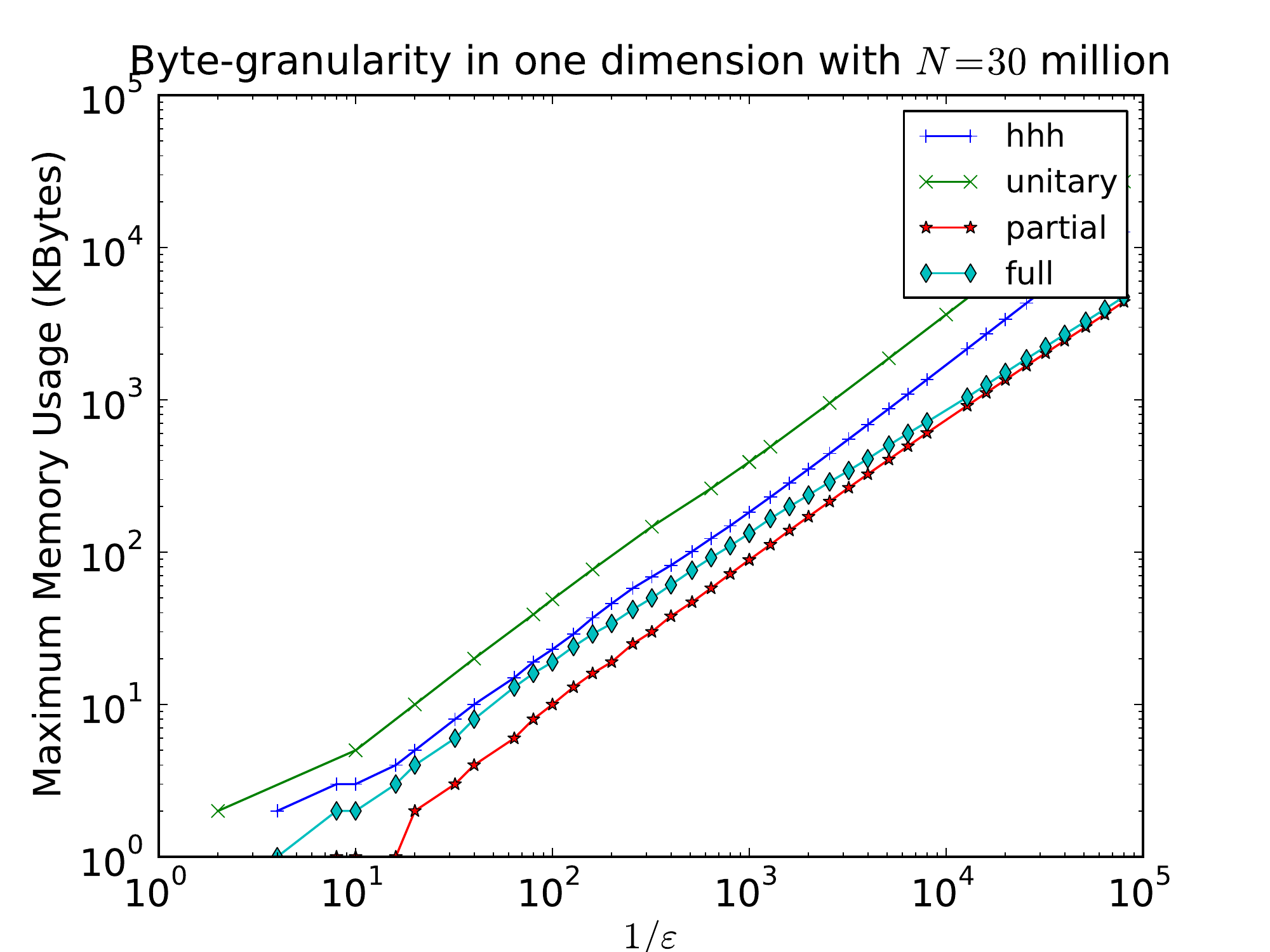}
\label{figure:vepsmemory_1} 
}%
\\
\subfloat[Memory usage in two dimensions for fixed stream length.]
{
\includegraphics[width=\figwidth]{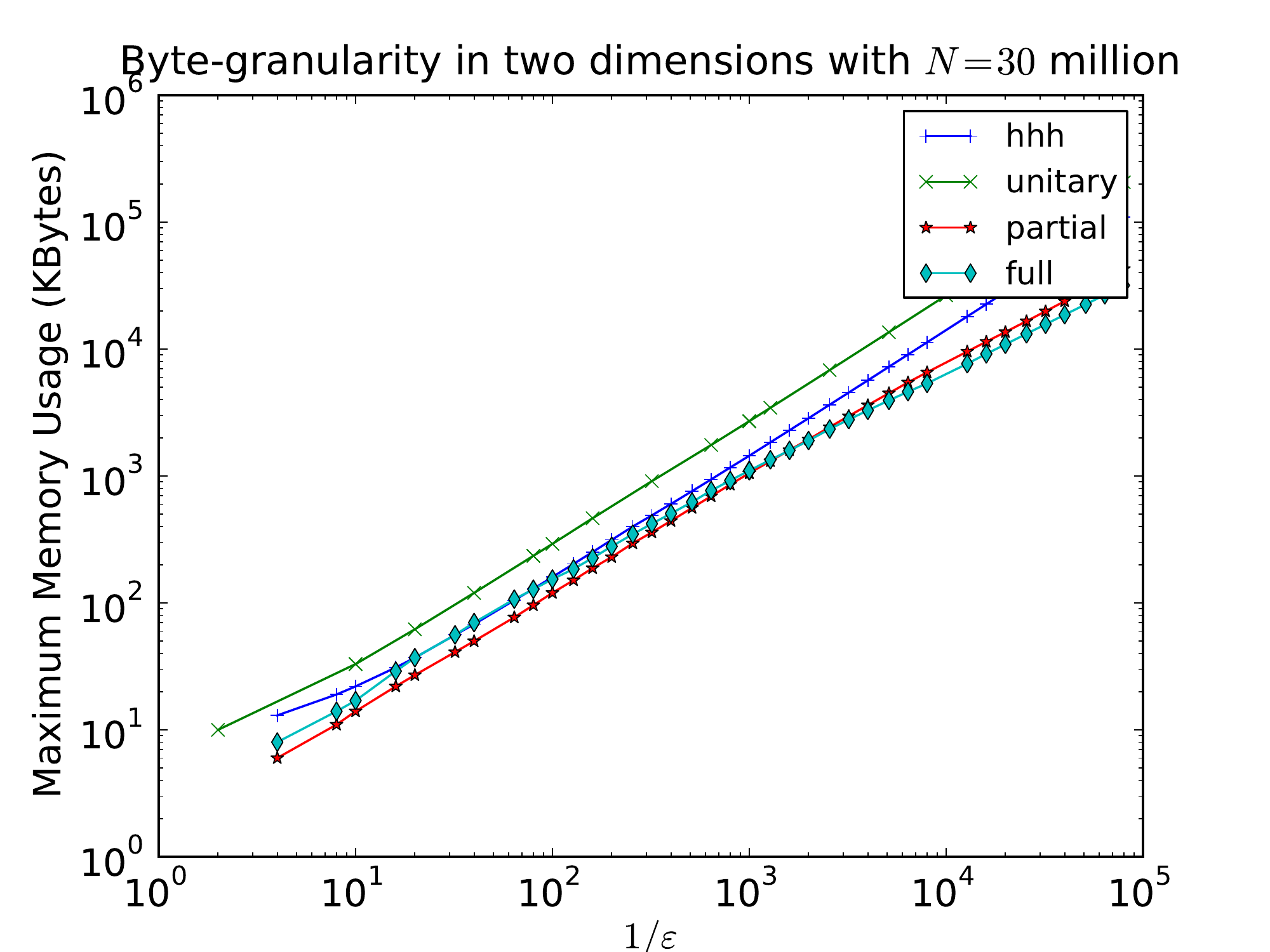}
\label{figure:vepsmemory_2}
}%
\hspace{0.08in}
\subfloat[Speed comparison in two dimensions with high $\epsilon$.]
{
\includegraphics[width=\figwidth]{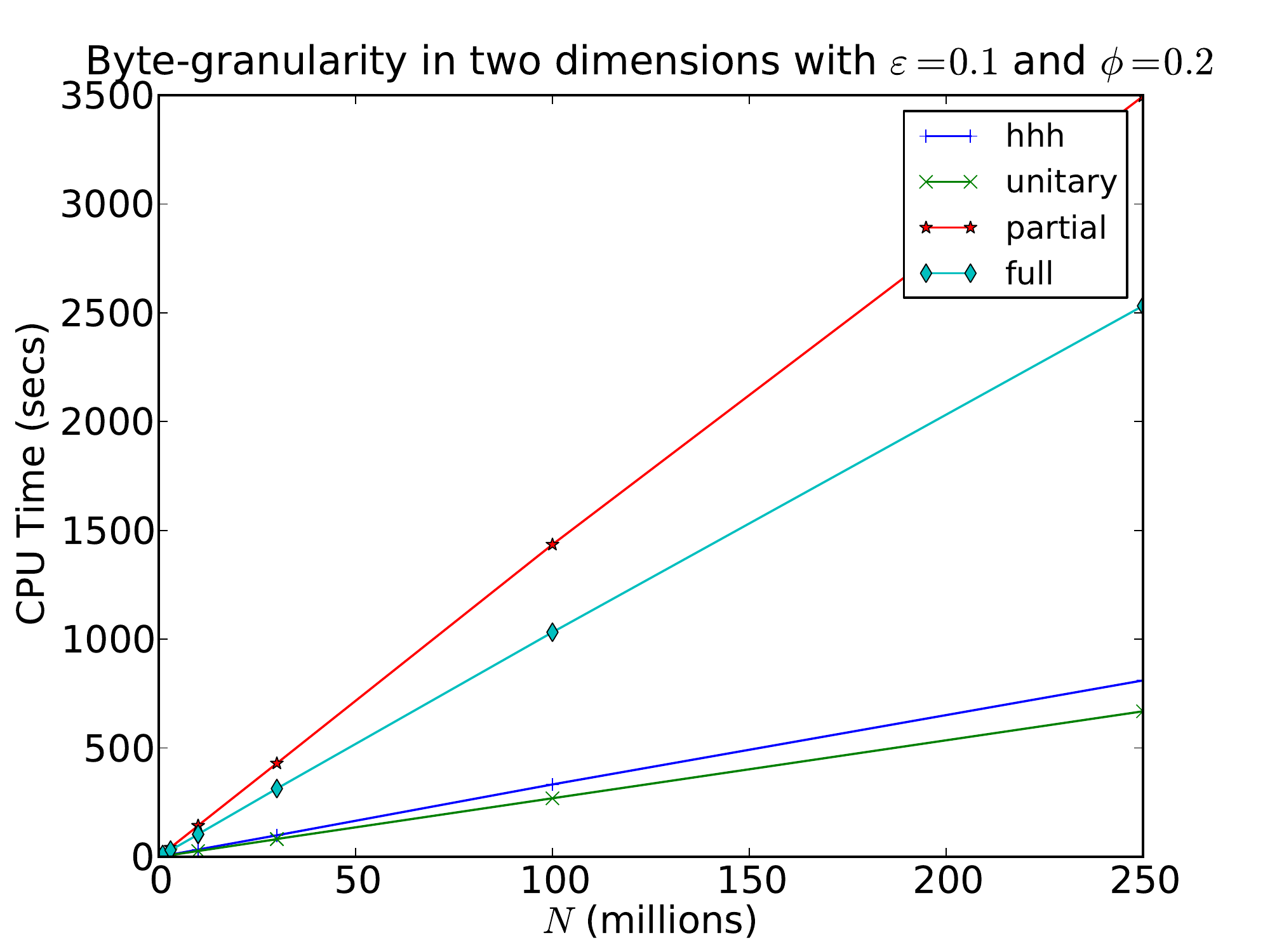}
\label{figure:time_10-5_2}
}%
\hspace{0.08in}
\subfloat[Speed comparison in two dimensions with medium $\epsilon$.]
{
\includegraphics[width=\figwidth]{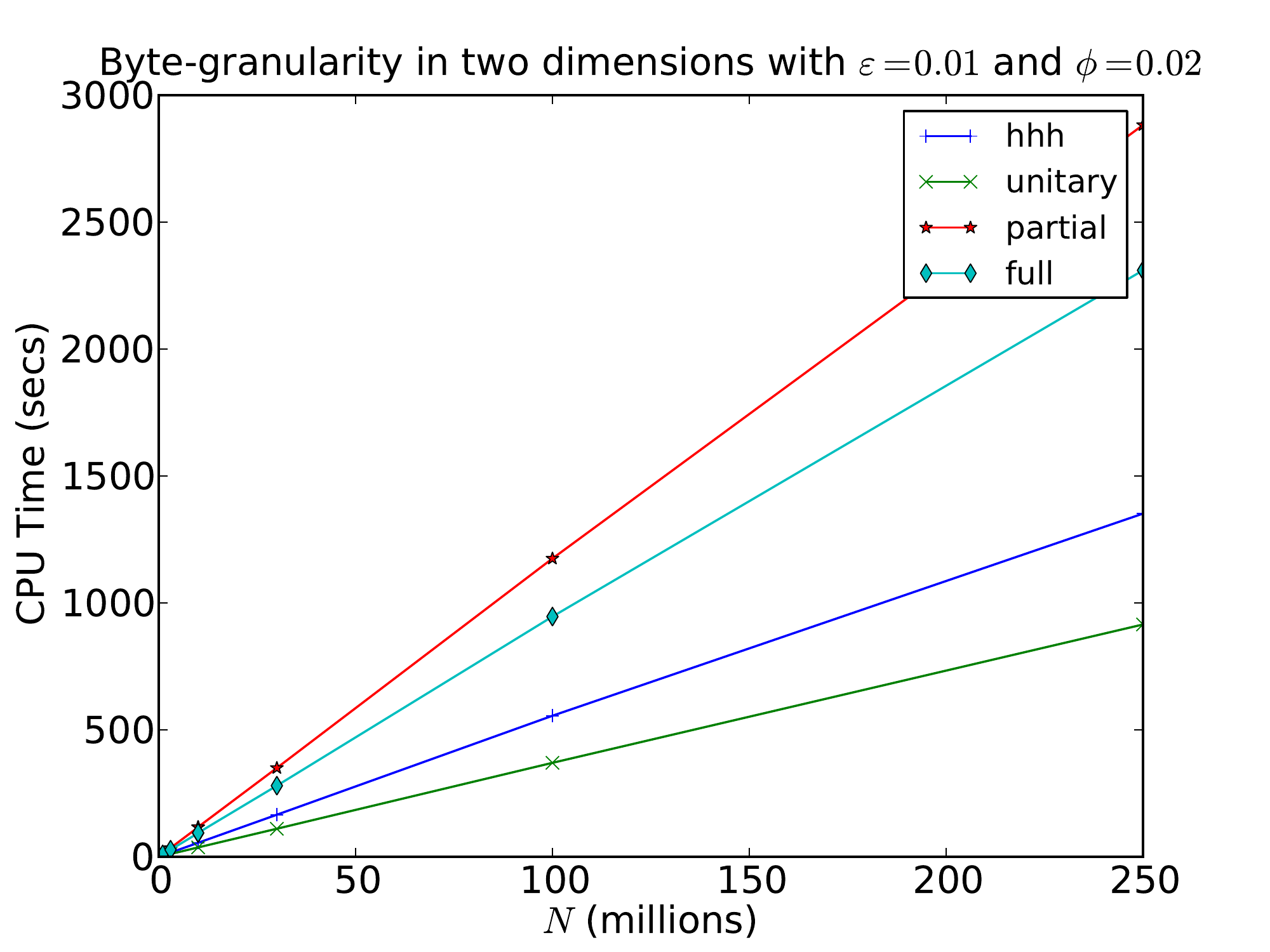}
\label{figure:time_10-5_1}
}\\
\subfloat[Speed comparison in two dimensions with low $\epsilon$.]
{
\includegraphics[width=\figwidth]{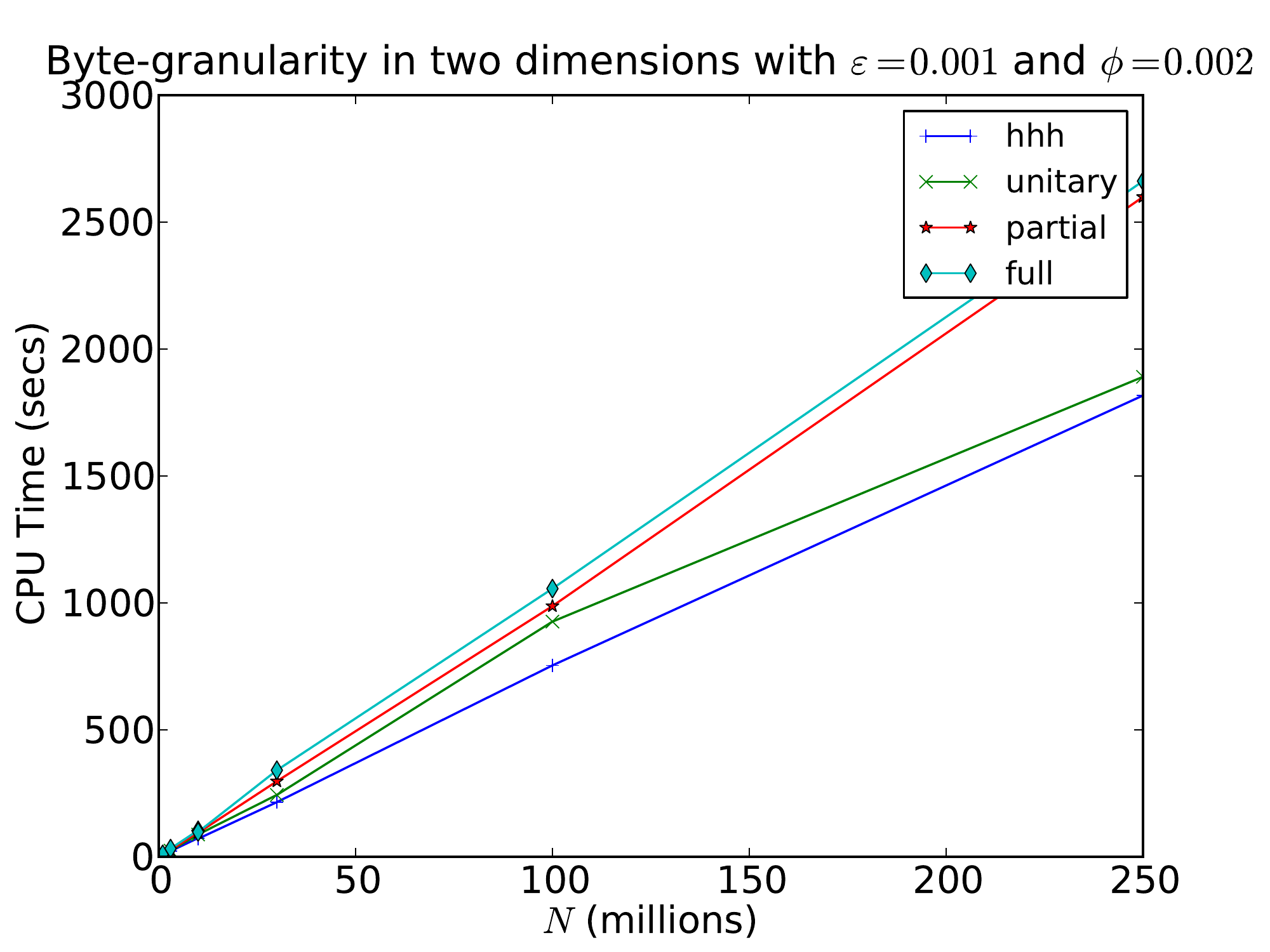}
\label{figure:time_1000-500_2}
}%
\subfloat[Speed comparison in one dimension with high $\epsilon$.]
{
\includegraphics[width=\figwidth]{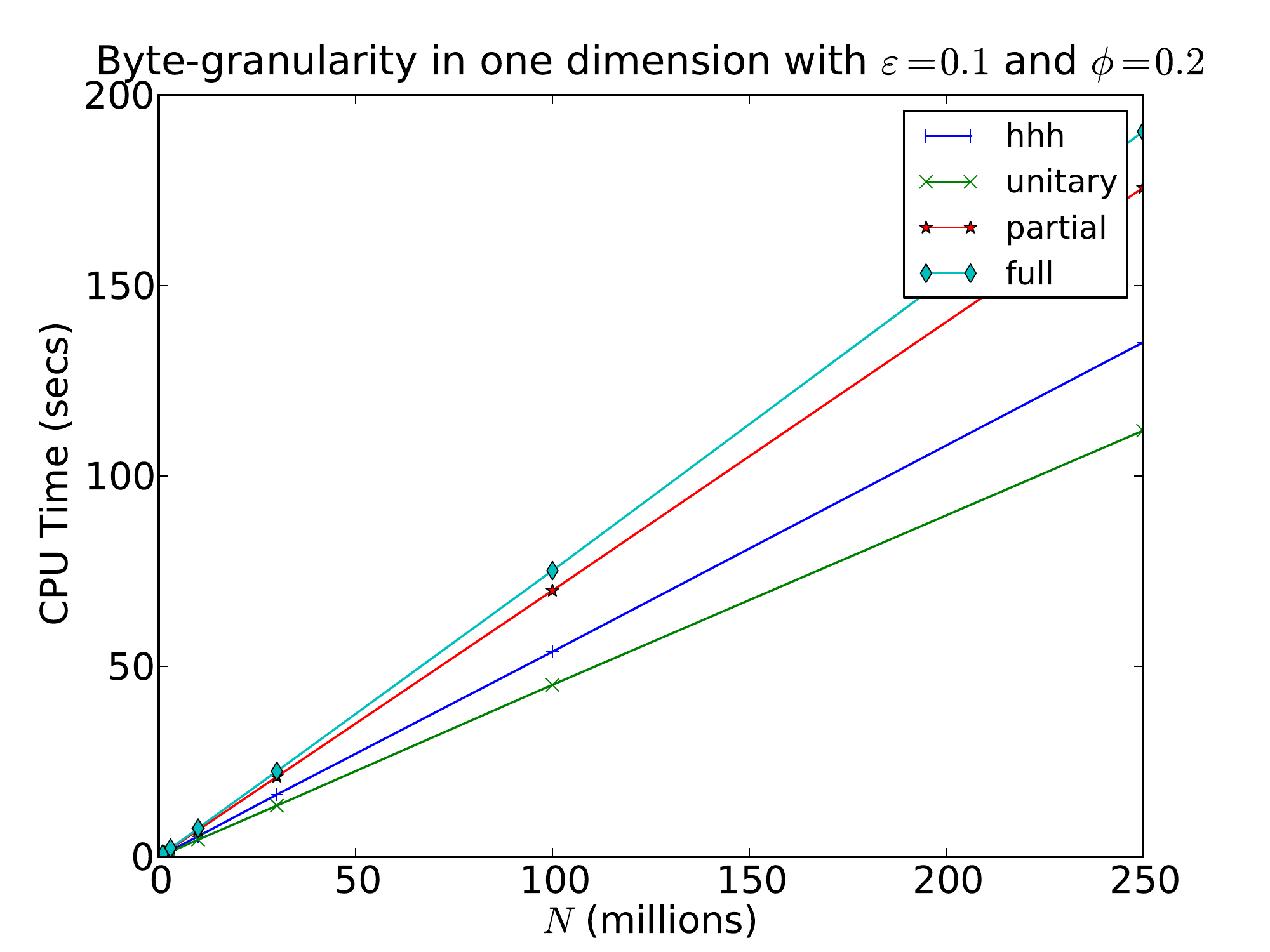}
\label{figure:time_10-5_1}
}%
\hspace{0.08in}
\subfloat[Speed comparison in one dimension with medium $\epsilon$.]
{
\includegraphics[width=\figwidth]{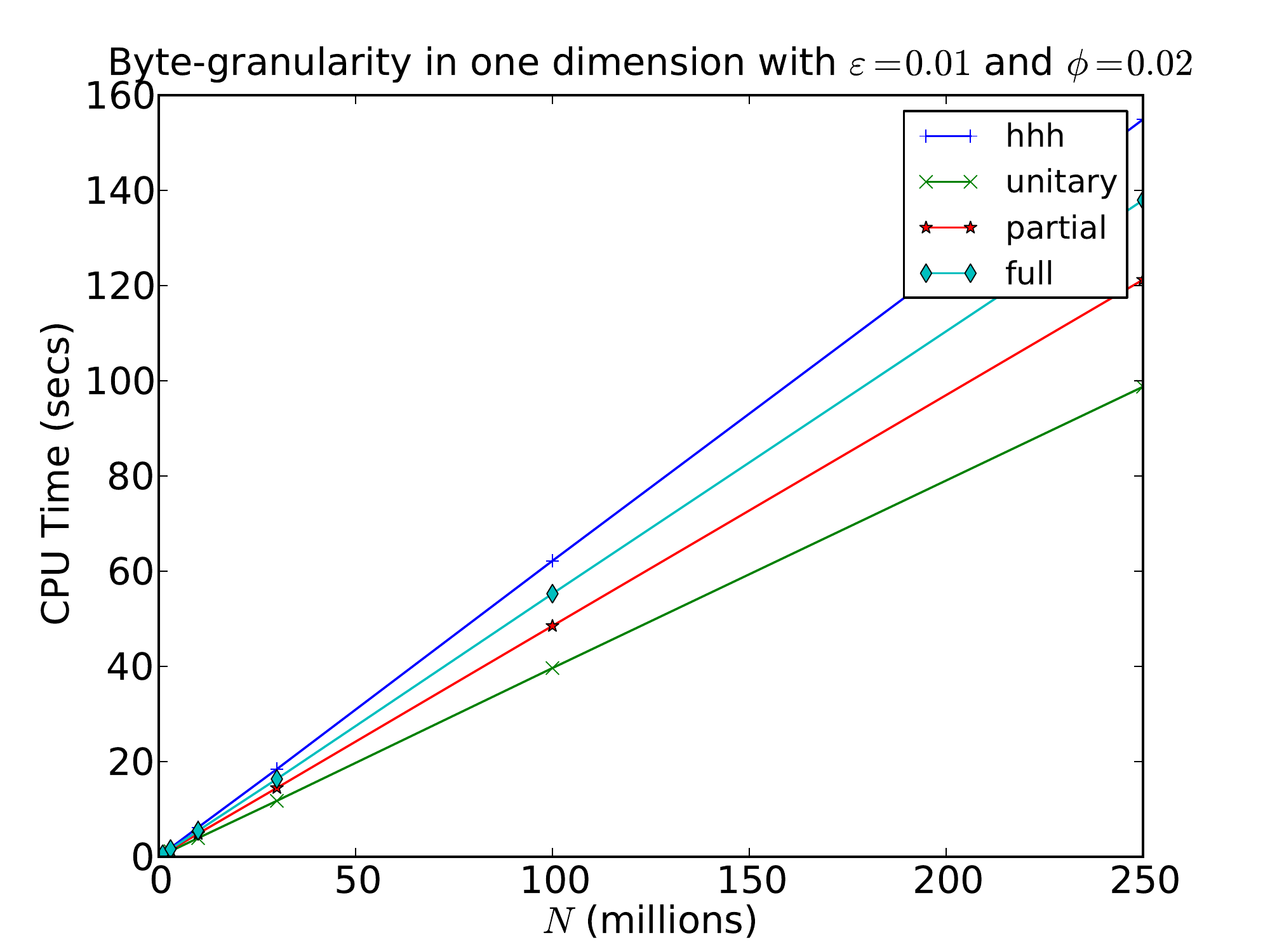}
\label{figure:time_100-50_1}
}
\\
\subfloat[Speed comparison in one dimension with low $\epsilon$.]
{
\includegraphics[width=\figwidth]{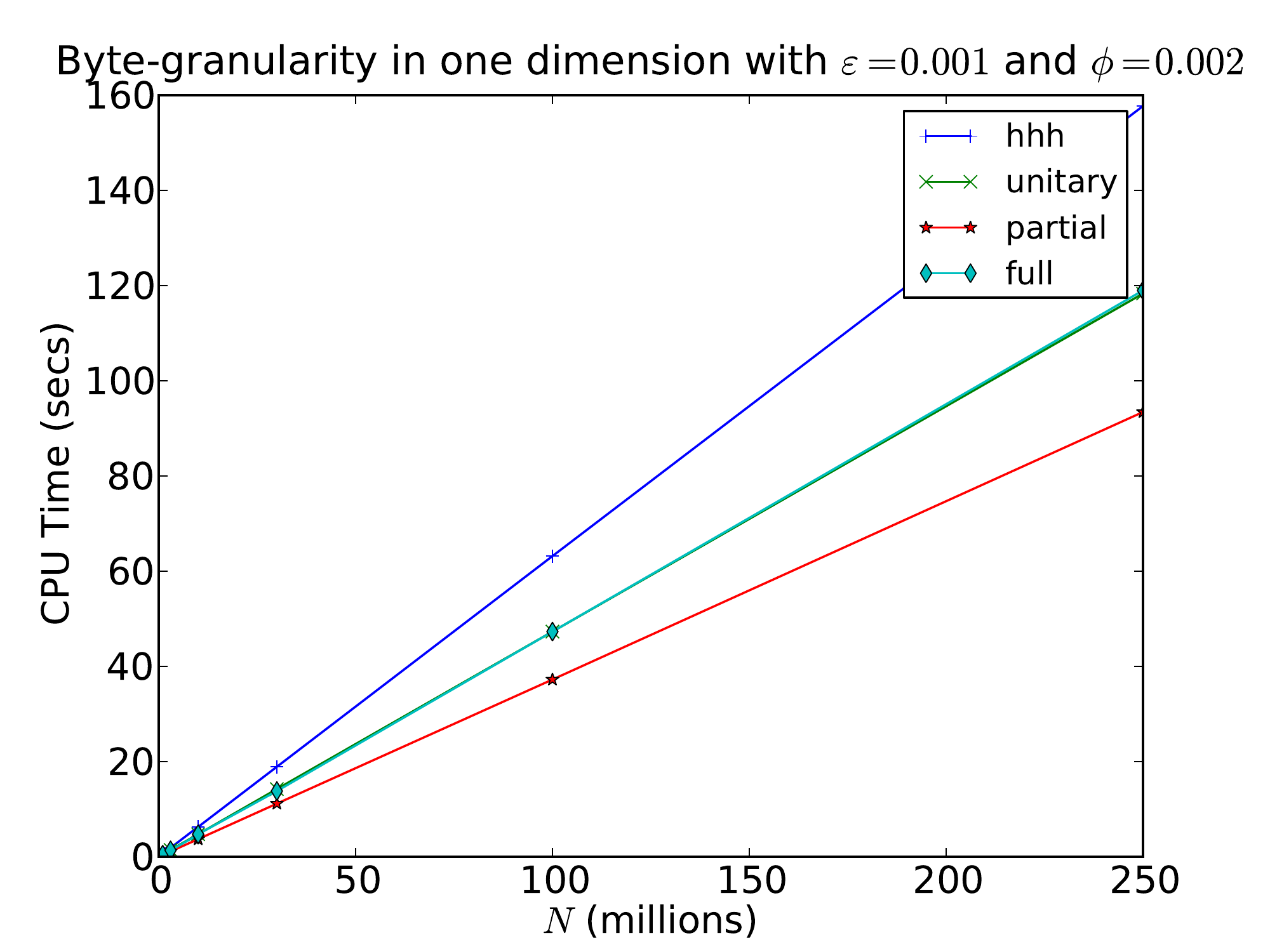}
\label{figure:time_1000-500_1}
}%
\hspace{0.08in}
\subfloat[Speed comparison in one dimension for fixed stream length.]
{
\includegraphics[width=\figwidth]{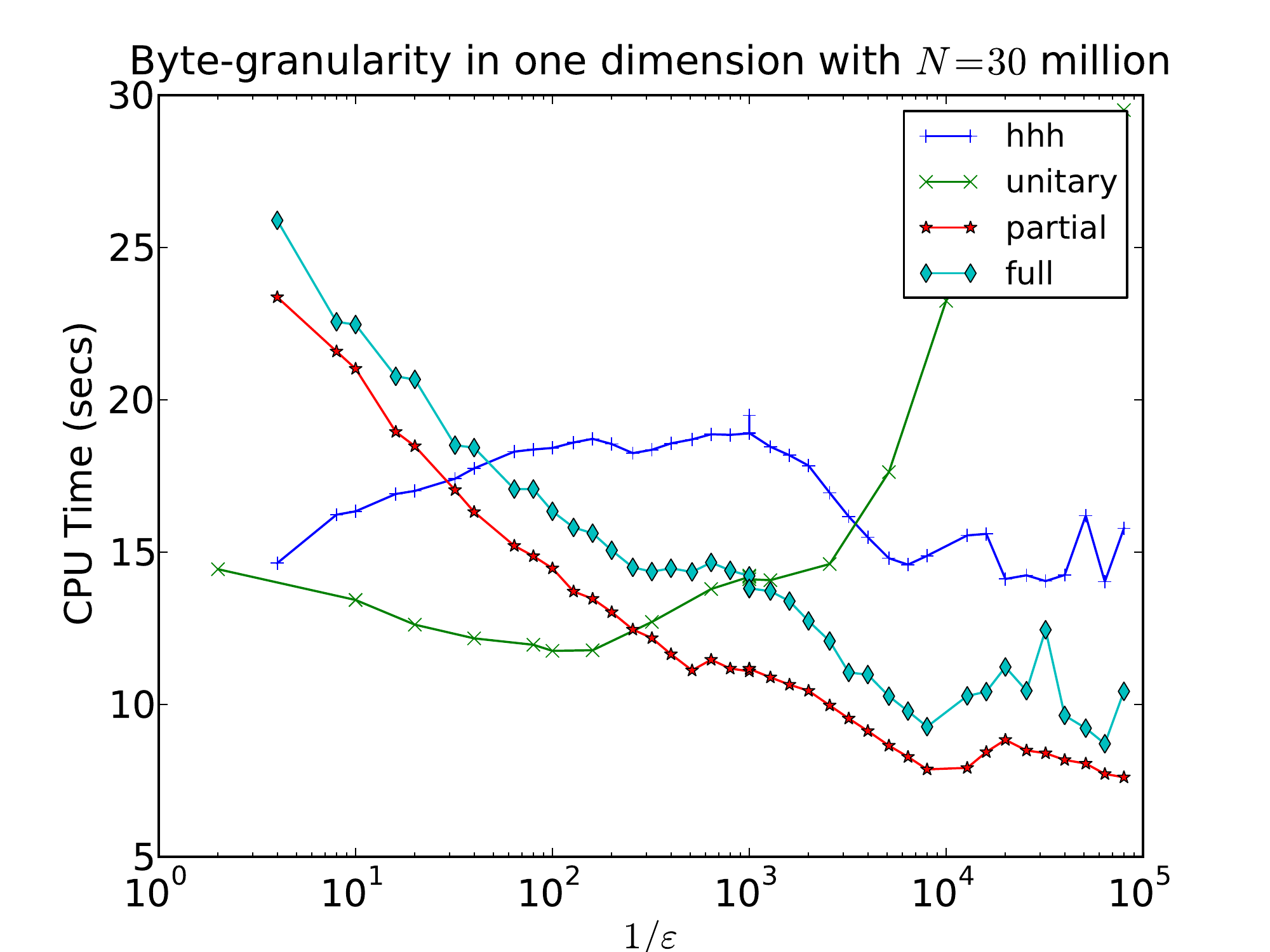}
\label{figure:vepstime_1}
}%
\hspace{0.08in}
\subfloat[Speed comparison in two dimensions for fixed stream length.]
{
\includegraphics[width=\figwidth]{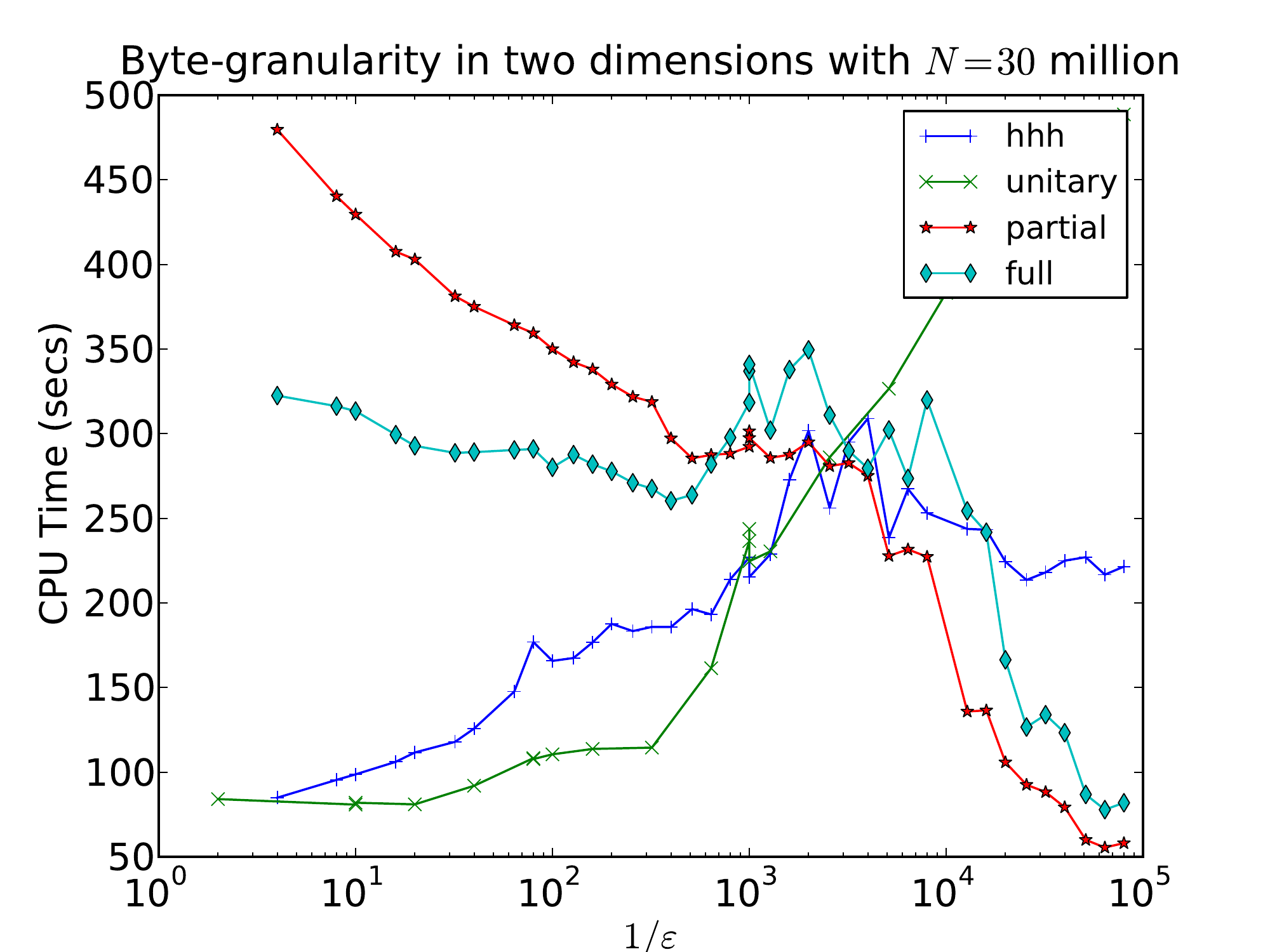}
\label{figure:vepstime_2}
}%
\caption{Memory and speed comparisons.}
\end{figure*}

\eat{
\begin{figure}[t]
\centering
\includegraphics[width=\figwidthes]{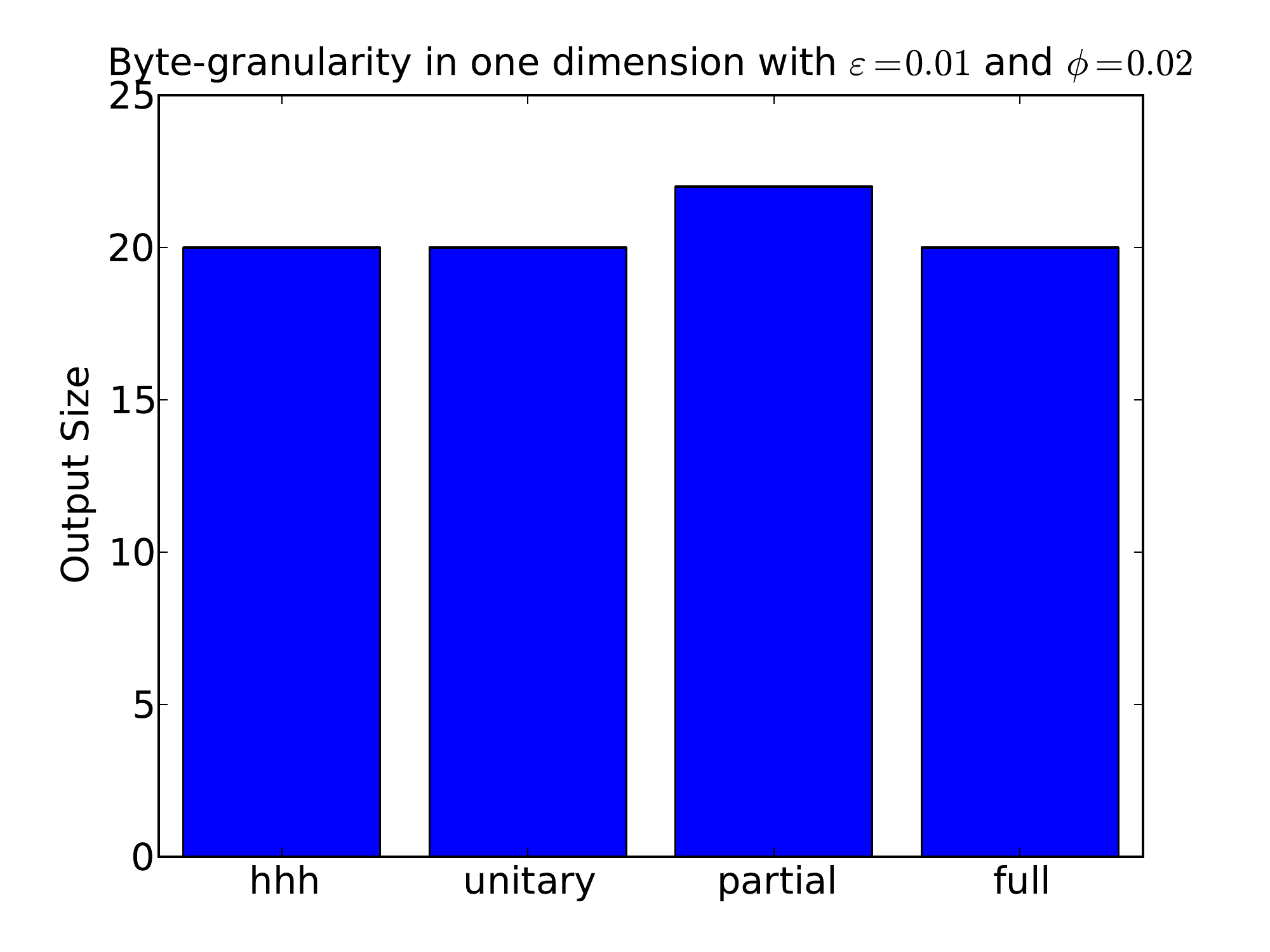}
\vspace{-2mm}
\caption{Output size comparison. For each algorithm, we display the maximum output size over all streams tested (for the given setting of $\epsilon$ and $\phi$).}
\label{figure:outputsize_100-50_1}
\hspace*{-6mm}
\end{figure}
}

\begin{figure*}
\subfloat[Output size comparison. For each algorithm, we display the maximum output size over all streams tested (for the given setting of $\epsilon$ and $\phi$).]
{
\includegraphics[width=\figwidth]{outputsize_100-50_1}
\label{figure:outputsize_100-50_1}
}%
\hspace{0.08in}
\subfloat[Accuracy comparison in one dimension.]
{
\includegraphics[width=\figwidth]{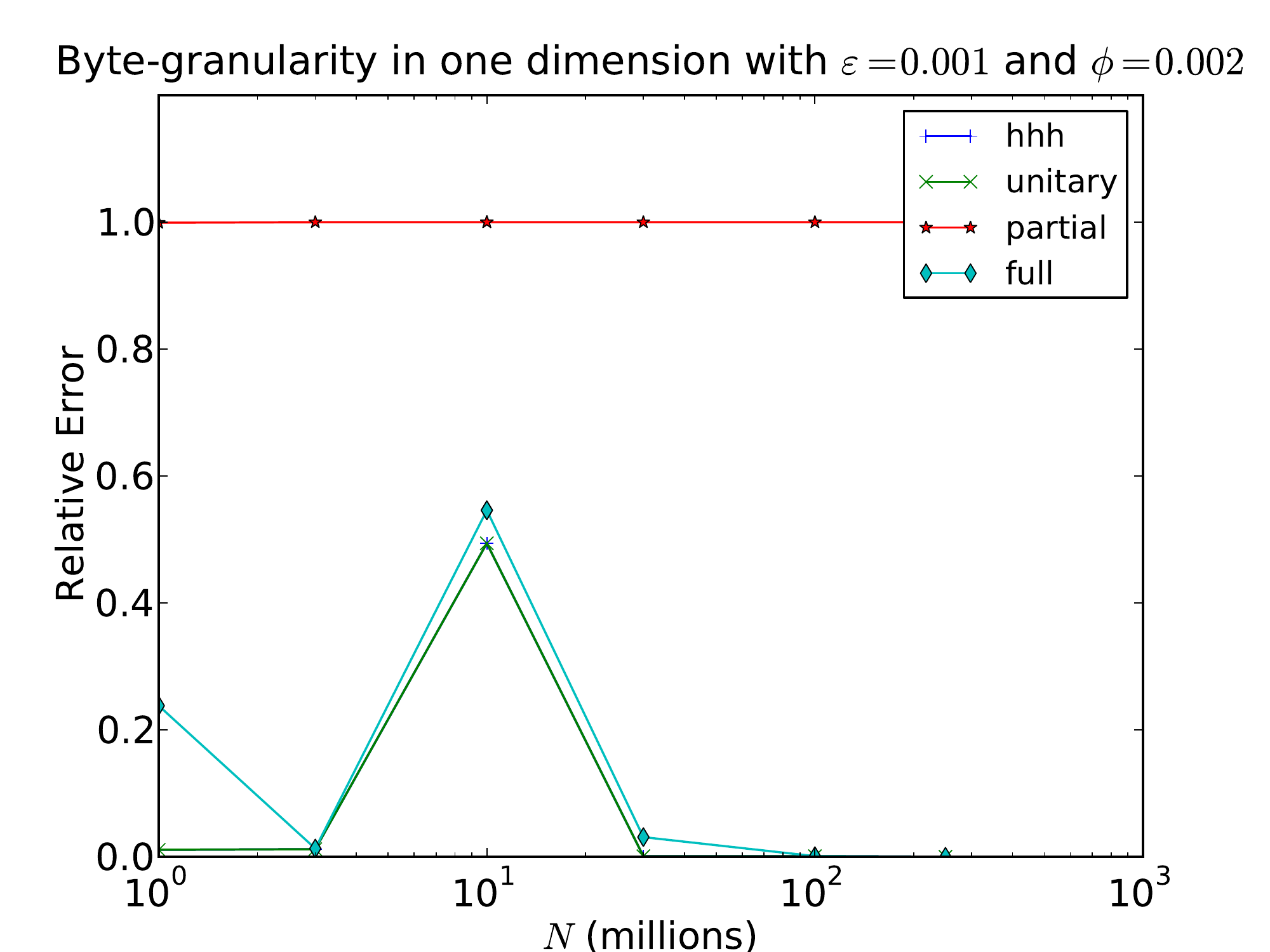}
\label{figure:epsilon_1000-500_1}
}%
\hspace{0.08in}
\subfloat[Accuracy comparison in two dimensions.]
{
\centering
\includegraphics[width=\figwidth]{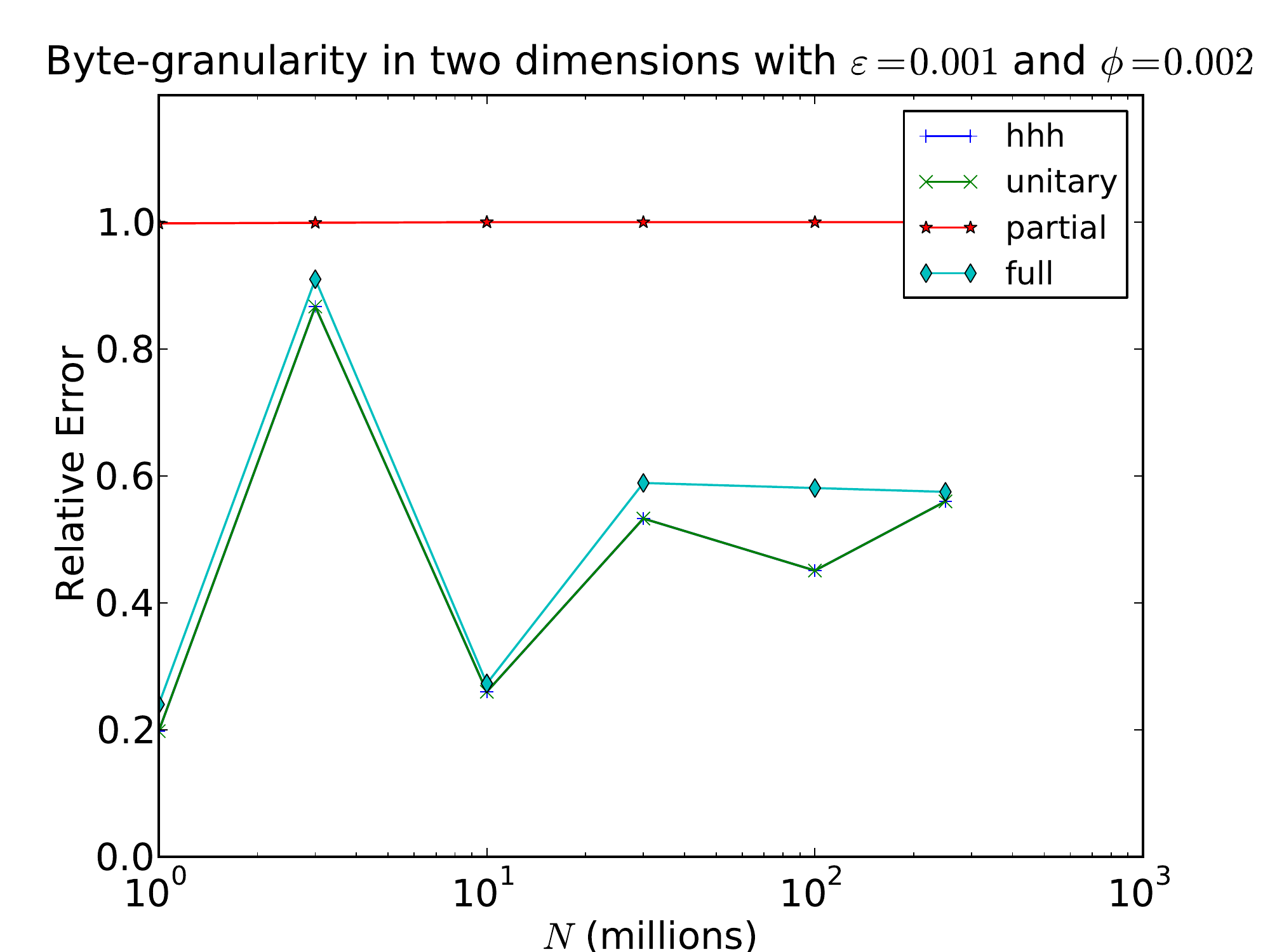}
\label{figure:epsilon_1000-500_2}
}%
\caption{Accuracy and output size comparisons.}
\end{figure*}


\medskip
\noindent \textbf{Memory.}
Both versions of our algorithm use more memory than \texttt{full} and \texttt{partial}.  The difference between \texttt{hhh}, \texttt{partial},
and \texttt{full} is a small constant factor; \texttt{unitary} uses about twice as much space
as \texttt{hhh}.  The largest difference in space usage appears in one
dimension, as shown in Figure
\ref{figure:memory_1000-500_1}. The difference is much
smaller in two dimensions, as shown in Figure
\ref{figure:memory_100-50_2}. In both cases, the better space usage of
\texttt{partial} comes at the cost of significantly decreased accuracy and increased
output size, as discussed below. We conclude that in situations where the decreased accuracy of \texttt{partial}
cannot be tolerated, the memory usage of our algorithms is not a major disadvantage, as \texttt{hhh} and 
\texttt{full} have similar memory requirements, especially in two dimensions.

Ideally, we would be able to present our results with the independent
variable a programmer-level object such as memory usage, rather than
the error-parameter $\epsilon$.  In practice, a programmer may be
allowed at most 1 MB of space to deploy an HHH algorithm on a network
sensor, and have to optimize speed and accuracy subject to this
constraint.  But while the mapping between $\epsilon$ and memory usage
is straightforward for our algorithm (the Space Saving implementation
uses 36 bytes per counter, and we use a fixed $\frac{H}{\epsilon}$
counters), this mapping is less clear for \texttt{partial} and
\texttt{full}, as their space usage is data dependent, with counters
added and pruned over the course of execution. Figures
\ref{figure:vepsmemory_1} and \ref{figure:vepsmemory_2} show the empirical mapping
between space usage and $\epsilon$ for a fixed stream length of
$N=30$ million with one- and two-dimensional bytewise hierarchies.
This setting highlights the importance of our improved worst-case
space bounds, even though our algorithm uses slightly more space in
practice.  It can be imperative to \emph{guarantee} assigned memory
will not be exceeded, and our algorithm allows a more aggressive
choice of error parameter while maintaining a worst-case
guarantee.

We emphasize that we did {\em not} attempt to optimize memory usage for our
algorithms using characteristics of the data, as suggested in Theorem~\ref{thm1}.
It is therefore likely that our algorithms can function with less memory than
\texttt{partial} and
\texttt{full} in many practical
settings.


Note that the running time and memory usage are independent of $\phi$, as $\phi$ only affects the output stage, which we have not included in our measurements as the resources consumed by this stage were negligible.

\vspace{-1mm}
\medskip
\noindent \textbf{Time.}
We observe that in both one and two dimensions, both \texttt{unitary} and \texttt{hhh} are faster than \texttt{partial} and \texttt{full} except for extremely small values of $\epsilon$. The speed of each algorithm for each setting of $\epsilon$ is illustrated for a fixed stream length of $N=30$ million in Figures
\ref{figure:vepstime_1} and \ref{figure:vepstime_2}; our algorithms are fastest in one dimension for $\epsilon$ greater than about $.01$ and in two dimensions for $\epsilon$ greater than about $.0001$. 

We show how runtime grows with stream length for fixed values of $\epsilon$ in Figures \ref{figure:time_10-5_2}
-\ref{figure:time_1000-500_1}. For concreteness, on a stream with $N=250$ million, \texttt{unitary} processes about $2.2$ million updates per second in one dimension at a byte-wise granularity when $\epsilon=.1$, while \texttt{hhh} processes $1.85$ million, \texttt{partial} $1.3$ million, and \texttt{full} processes $1.4$ million. Here, $N$ corresponds to the number of packets (not weighted by size), and the updates per second statistic specifies the number of packets processed per second by our implementation. In two dimensions for $\epsilon=.1$,  \texttt{unitary} processes over $370,000$ updates per second and \texttt{hhh} processes $300,000$, while \texttt{partial} processes $71,000$, and \texttt{full} processes about $100,000$. Thus, our algorithms ran more than three times faster than \texttt{partial} and \texttt{full} for this particular setting of parameters. 



\vspace{-2mm}
\medskip
\noindent \textbf{Output Size.}
\label{sec:outputsize}
The output size and accuracy are measures of the quality of output and depend on the value of $\phi$. All three algorithms produce a near-optimal output size, with \texttt{partial} consistently outputting the largest sets. The largest difference observed is shown in Figure \ref{figure:outputsize_100-50_1}.

\vspace{-1mm}
\medskip
\noindent \textbf{Accuracy.}
\label{sec:accuracy}
We define the relative error of the output to be $\max_{p \in \text{output}} \frac{f_\text{max}(p)-f_\text{min}(p)}{\epsilon N}.$ Clearly, the relative error is between 0 and 1, because of the accuracy guarantees of our algorithms.  We find that the relative error can vary significantly for all algorithms, but our algorithm uniformly performs best. The relative error of the partial ancestry algorithm is often close to the theoretical upper bound of 1, making it by far the least accurate of the algorithms tested.

\vspace{-1mm}
\medskip
\noindent \textbf{TCAM Simulations.} We simulated our TCAM-conscious implementation of our algorithm on the same packet traces as above, in order to estimate the number of TCAM operations our implementation requires per packet processed. \cite{tcam} experimentally demonstrates that TCAM READ, WRITE, and SEARCH operations all take roughly the same amount of time. Thus, we counted the total number of READ, WRITE, and SEARCH operations our TCAM implementation required,
without distinguishing between the three. We found that for one-dimensional IP addresses at byte-wise granularity, each packet required about 14 TCAM operations on average, or 2.8 TCAM operations per instance of Space Saving maintained by our algorithm. This is slightly better than the worst-case behavior of the implementation of \cite{tcam}, which requires up to 4 TCAM operations per update. Our two-dimensional algorithm at byte-wise granularity requires about 65 TCAM operations per packet; since the two-dimensional algorithm maintains 25 instances of Space Saving, this translates
to only 2.6 TCAM operations per instance of Space Saving. We attribute this improvement in TCAM operations per instance of Space Saving to the fact that the frequency distribution at high nodes in the two-dimensional lattice is highly non-uniform.

\vspace{-2mm}


\section{Conclusion}
The trend in the literature on the approximate HHH problem has been towards increasingly complicated algorithms. In this work, we present what is perhaps the simplest algorithm for HHHs in arbitrary dimension, and demonstrate that it is superior to the existing standard in many respects, and competitive in all others. We believe our algorithm offers the best tradeoff between simplicity and performance.

\medskip
\noindent \textbf{Acknowledgements.} We thank Elaine Angelino, Richard Bates, and Graham Cormode for helpful discussions. 
Michael Mitzenmacher is supported by NSF grants CNS-0721491, CCF-0915922, and IIS-0964473, and grants from Cisco, Inc., Google, and Yahoo!.
Justin Thaler is supported by the 
Department of Defense (DoD) through the National Defense Science \&
Engineering Graduate Fellowship (NDSEG)  
Program.

\vspace{-1mm}
{

}

\eat{

}
\normalsize
\vspace{-3mm}
\appendix
\section{Proof of Theorems}

\label{app:proof}

\begin{proof}[Proof of Theorem~\ref{thm:multid}] First note that it is possible, using the inclusion-exclusion principle, to show that $$F_p = f(p)-\sum_{h \in H_p} f(h) + \sum_{(h, h'\in H_p) \wedge q = \text{glb}(h, h')} f(q) $$
\vspace{-2mm}
$$- \sum_{(h, h', h''\in H_p) \wedge q = \text{glb}(h, h', h''')} f(q) + \dots $$ 

 We claim that for all $u$ expressible as the greatest lower
bound of more than two elements of $H_p$, the total contribution of
$f(u)$ to the above sum is 0. Indeed, suppose that $u = (u_1,u_2)$
is a descendant from exactly $m$ such elements, $h_1, h_2,\dots,h_m$
in $H_p$.  Since $u \prec h_{\alpha}$ for $\alpha$ in $\{1, \dots, m\}$ these
$m$ heavy hitter elements can be written as $h_1=(P_{i_1}u_1,
  P_{j_1}u_2),h_2=(P_{i_2}u_1, P_{j_2}u_2),\dots,h_m=(P_{i_m}u_1,
  P_{j_m}u_2)$, where $(P_i u_1, P_j u_2)$ denotes the element
obtained from $(u_1, u_2)$ by generalizing $i$ times on the first
dimension and $j$ times on the second dimension. Renumbering if necessary,
assume the nodes
are sorted on the generality of their first
component so that $i_1 < i_2 < \dots < i_m$.  It is clear that there
are no equalities in the sequence because if $ i_{\alpha}= i_{\beta} $
then either $h_{\alpha} \preceq h_{\beta}$ or $h_{\beta} \preceq
h_{\alpha}$ which contradicts that these are from $H_p$.  When the corresponding
relationships between the second components are examined it can be
seen that the increasing sort on the first component forces a
decreasing order on the second component $j_1 > j_2 > \dots > j_m$
(since if $\alpha < \beta $ and $j_{\alpha} \le j_{\beta}$ then
because $i_{\alpha} < i_{\beta}$ the contradiction $h_\alpha \prec
h_{\beta} $ is reached). Thus the $m$ elements are in a linear
structure with endpoints $h_1$ and $h_m$. Clearly the first component
of $u$ is the first component of $h_1$ and similarly the second
component of $u$ is the second component of $h_m$.  
With this in hand
it is clear that $u$ is the greatest lower bound of any subgroup of
the $m$ elements that includes $h_1$ and $h_m$.  There are $\binom
{m-2} k $ ways to pick $k$ middle terms and thus there are $\binom {m-2}
k $ ways in which node $u$ appears as the greatest lower bound of
$k+2$ elements from $H_p$.  Returning to the sum $$F_p=f(p)-\sum_{h
  \in H_p} f(h) + \sum_{(h, h'\in H_p) \wedge q = \text{glb}(h, h')} f(q)$$
    \vspace{-2mm}
  $$
- \sum_{(h, h', h''\in H_p) \wedge q = \text{glb}(h, h', h''')} f(q) +
\dots $$

\noindent it is now clear that $f(u)$ will appear once in the sum over pairs, $m-2$ times in the sum over triples, and, in general, $\binom {m-2} k $ times in the sum over groups of size $k+2$.  When combined with the sign structure in the sum this gives a resulting contribution from $u$ of $$f(u) \sum_{j=0}^{m-2}(-1)^j \binom {m-2} j = f(u)(1-1)^{m-2} = 0. $$  Thus in the two dimensional case

$$F_p=f(p)-\sum_{h_1 \in H_p} f(h_1) + \sum_{q \in T_p} f(q)$$ 

\noindent as claimed. \end{proof}

\begin{proof}[{Proof of Theorem \ref{thm:multidbound}}]
The proof will closely parallel that of Theorem \ref{thm:conditioned}. We bound the total error in the estimated conditioned  counts, aggregated over all $p \in P$, and this will imply that the sum of the true conditioned  counts of all $p \in P$ is large. Hence there cannot be too many approximate HHHs output. 

We showed in Theorem \ref{thm:multid} that $$F_p=f(p)-\sum_{h_1 \in H_p} f(h_1) + \sum_{q \in T_p} f(q).$$

Therefore, $$F'_p - F_p =\big(f_{\text{max}}(p) -f(p)\big) + \sum_{h_1 \in H_p} \big(f(h_1) - f_{\text{min}}(h_1)\big)$$
$$ + \sum_{q \in T_p} \big(f_{\text{max}} (q) - f(q)\big).$$

Our goal is to show that the sum of the true conditioned  counts of all $p \in P$ is large by bounding the total error in the estimated conditioned  counts, aggregated over all $p \in P$. To this end, consider the sum 
$$\sum_{p \in P} F_p = \sum_{p \in P} F'_p - \sum_{p \in P} (F'_p - F_p) \geq |P| \phi N - \sum_{p \in P} (F'_p - F_p)$$
$$=|P| \phi N - \sum_{p \in P} \big(f_{\text{max}}(p) -f(p)\big) - $$
$$\sum_{p \in P}  \sum_{h_1 \in H_p} \big(f(h_1) - f_{\text{min}}(h_1)\big) - \sum_{p \in P}  \sum_{q \in T_p} \big(f_{\text{max}} (q) - f(q)\big).$$

We refer to the second term on the right hand side of the last expression, $\sum_{p \in P} \big(f_{\text{max}}(p) -f(p)\big)$, as ``Term-Two error", the third term, $\sum_{p \in P}  \sum_{h_1 \in H_p} \big(f(h_1) - f_{\text{min}}(h_1)\big)$, as ``Term-Three" error, and the fourth term, $\sum_{p \in P}  \sum_{q \in T_p} \big(f_{\text{max}} (q) - f(q)\big)$ as ``Term-Four error". By the Accuracy guarantees, it is immediate that the Term-Two error is bounded above by $|P|\epsilon N$. 

In order to bound Term-Three error, we must briefly introduce the notion of comparable items in a lattice. Two elements $x$ and $y$ are \emph{comparable} under the $\preceq$ relation if the label of $y$ 
is less than or equal to that of $x$ on every attribute. Let $A$ be the size of the largest antichain in the lattice, that is, the maximum size of any subset of prefixes such that any two items in the subset are incomparable. It was shown in \cite{journal} that $A = 1+\min(h_1, h_2)$. 
We show that  $\sum_{p \in P} |H_p| \leq A|P|$; it then follows by the Accuracy guarantees that the Term-Three error is bounded above by $|P| A \epsilon N$. To this end, for any $h \in P$, consider the set $B_h=\{p \in P: h \in H_p\}$. We claim that  $|B_h| \leq A$, since all the items in $B_h$ must be incomparable. Indeed, suppose $p, q \in B_h$ and the label of $q$ is less than the label of $p$ on both attributes. Then $h \prec q \prec p$, so by definition of $H_p$, $h \not\in H_p$, which is a contradiction. Thus, $\sum_{p \in P} |H_p|\!=\!\sum_{h \in P} |B_h| \leq A|P|$. 

Finally, we may bound the Term-Four error by $\frac{AP^2}{4} \epsilon N$. This will clearly follow from the Accuracy guarantees if we can bound $\sum_{p \in P} |T_p|$ by $\frac{A|P|^2}{4}$. To this end, for each $p \in P$ let $G_p$ be a graph on $|H_p|$ vertices, where edge $(h_1, h_2) \in E(G_p)$ if and only if $\text{glb}(h_1, h_2) \in T_p$. It is clear that $|T_p| = |E(G)|$. We claim $G$ is a triangle-free graph -- it then follows by Turan's theorem \cite{turan} that $|T_p| \leq \frac{|H_p|^2}{4}$. For three distinct vertices $h_1, h_2, h_3 \in H_p$, let $u_1=\text{glb}(h_1, h_2)$, $u_2=\text{glb}(h_2, h_3)$ and $u_3 = \text{glb}(h_1, h_3)$. We show that if $u_1$, $u_2$ and $u_3$ are all in $T_p$, then for at least one $i$, $u_i$ is a descendant of $h_1, h_2$, and $h_3$, contradicting $u_i \in T_p$.

Write $h_i=(h_{i, 1}, h_{i,2})$ for each $i \in \{1, 2, 3\}$. By assumption $(h_i, h_j)$ share a common descendant for any pair $(i, j)$, so we may assume (renumbering if necessary) that $h_{1, 1} \prec h_{2, 1} \prec h_{3, 1}$ as one-dimensional objects. The remainder of the proof now closely parallels that of Theorem \ref{thm:multid}.
 It is clear that there
are no equalities in the sequence because if $h_{i, 1}= h_{j, 1} $
then either $h_{i, 1} \preceq h_{j, 1}$ or $h_{j,1} \preceq
h_{i,1}$ which contradicts that these are from $H_p$.  For the same reason, it can be
seen that the increasing sort on the first component forces a
decreasing order on the second component, i.e., $h_{3, 2} \prec h_{2, 2} \prec h_{1, 2}$. 
Consequently, $u_3=glb(h_1, h_3) = (h_{1, 1}, h_{3, 2})$ is a descendant of  $h_1, h_2$, and $h_3$, contradicting $u_3 \in T_p$.

So we have shown that $|T_p| \leq \frac{|H_p|^2}{4}$. Since $\sum_{p \in P} |H_p| \leq A|P|$, and trivially $|H_p|\leq |P|$ for all $p \in P$, Holder's Inequality implies that $\sum_{p \in P} |T_p| \leq  \sum_{p \in P} \frac{|H_p|^2}{4} \leq A|P|^2/4$. 

\eat{ Indeed,  write  $h_1=(P_{i_1}u_1,
  P_{j_1}u_2),h_2=(P_{i_2}u_2, P_{j_2}u_3)$,  where as in the proof of Theorem \ref{thm:multid}, $(P_i u_1, P_j u_2)$ denotes the element
obtained from $(u_1, u_2)$ by generalizing $i$ times on the first
dimension and $j$ times on the second dimension. We showed in the proof of Theorem \ref{thm:multid} that without loss of generality
the first component of $u_1$ is the first component of $h_1$ and the second
component of $u_2$ is the second component of $h_1$. This means we may write $h_1=(a, d), u_1=(a, b), u_2=(c,d)$ for one-dimensional entries $a, b, c, d$ with $b$ a generalization of $d$
and $c$ a generalization of $a$.

Likewise, we may assume that the first component of $u_3$ is the first component of $h_2$ and the second
component of $u_2$ is the second component of $h_2$ (a symmetric argument will handle the case where 
the first component of $u_2$ is the first component of $h_2$ and the second
component of $u_3$ is the second component of $h_2$). This means we may write $h_2=(e, d), u_2=(c,d), u_3=(e, f), $ for some $e, f$ with
 $c$ a generalization of $e$ and $f$ a generalization of $d$.

Using the fact that $c$ is a generalization of both $a$ and $e$, we see that if $h_3 = \text{glb}(u_1, u_3)$ is non-trivial, then either
$a$ generalizes to $e$ or $e$ generalizes to $a$. Assume without loss of generality that $a$ generalizes to $e$. Then
$h_1=(a, d) \prec (e, d) = \text{glb}(u_2, u_3) \prec u_3$, so $h_i$ is a descendant of $u_1, u_2$, and $u_3$, contradicting $h_i \in T_p$.}

Thus, we have shown that $$\sum_{p \in P} F_p \geq |P| \phi N - |P| \epsilon N - A |P| \epsilon N - \frac{A|P|^2}{4} \epsilon N.$$

Now note that $A N \geq \sum_{p \in P} F_p$, because each fully-specified item $e$ can only contribute to the true conditioned  counts of incomparable approximate HHHs. For if $e$ contributes to the conditioned  count of both $p$ and $q$ then $e \preceq p \wedge e \not\preceq P_p$ and $e \preceq q \wedge e \not\preceq P_q$. If $p$ and $q$ are comparable, then this implies either $q \preceq p$ or $p \preceq q$, contradicting the fact that $e \not\prec P_p$ and $e \not\preceq P_q$. Thus, we see that  $AN \geq (\phi N - (A + 1) \epsilon N ) |P| - \frac{A|P|^2}{4} \epsilon N$.

Dividing through by $N$ and subtracting $A$ from both sides yields
\vspace{-2mm}
$$0 \geq -A + (\phi - (A+1) \epsilon)  |P| - \frac{A\epsilon}{4} |P|^2.$$

Using the quadratic equation,  this holds if and only if 

$$|P| \leq -2\frac{(1 + A) \epsilon - \phi + \sqrt{(\phi - (1 + A) \epsilon)^2 - A^2\epsilon}}{A\epsilon}$$ or  
\vspace{-2mm}
$$|P| \geq -2\frac{(1 + A) \epsilon - \phi - \sqrt{(\phi - (1 + A) \epsilon)^2 - A^2\epsilon}}{A\epsilon},$$

\vspace{-2mm}
and we can rule out the latter case for small enough $\epsilon$ via trivial upper bounds on $|P|$ such as $|P| \leq \frac{H}{\phi}$.
\end{proof}


\end{document}